\documentclass[3p]{elsarticle}
\usepackage[utf8]{inputenc}
\usepackage{amsfonts}
\usepackage[linesnumbered]{algorithm2e}
\usepackage{subfigure}
\usepackage{tikz}
\usepackage{hyperref}
\usepackage{ifthen}

\SetAlgoLined
\SetKw{KwLet}{let}
\SetKw{KwOr}{or}
\SetKwFor{ElseForEach}{else foreach}{do}{end}
\SetKwFor{WithProb}{with probability}{do}{end}

\newcommand{\reais}{\mathbb{R}}
\newcommand{\reaisnneg}{\mathbb{R}_+}
\newcommand{\recebe}{\leftarrow}
\newcommand{\Oh}{\mathrm{O}}

\newcommand{\Esp}{\mathbb{E}}
\newcommand{\Ical}{\mathcal{I}}
\newcommand{\Vcal}{\mathcal{V}}
\newcommand{\opt}{\ensuremath{\mathrm{opt}}}

\newtheorem{theorem}{Theorem}
\newtheorem{lemma}[theorem]{Lemma}
\newtheorem{fact}[theorem]{Fact}
\newdefinition{definition}[theorem]{Definition}
\newdefinition{hypothesis}[theorem]{Hypothesis}
\newdefinition{problem}[theorem]{Problem}
\newproof{proof}{Proof}

\journal{Theoretical Computer Science}

\begin{document}

\begin{frontmatter}

\title{Query-Competitive Sorting with Uncertainty}
\tnotetext[t1]{Partially supported by Icelandic Research Fund grant 174484-051.
A preliminary version of this paper appeared in volume~138 of LIPICs, article 7, 2019. DOI: 10.4230/LIPIcs.MFCS.2019.7}

\author[ru]{Magnús M. Halldórsson}
\ead{mmh@ru.is}

\author[le]{Murilo Santos de Lima\corref{cor}}
\ead{mslima@ic.unicamp.br}

\cortext[cor]{Corresponding author}

\address[ru]{Department of Computer Science, Reykjavik University, Menntavegi 1, 102 Reykjavik, Iceland}

\address[le]{School of Informatics, University of Leicester, University Road, Leicester, LE1 7RH, United Kingdom}

\begin{abstract}
We study the problem of sorting under incomplete information, when queries are used to resolve uncertainties.
Each of $n$ data items has an unknown value, which is known to lie in a given interval.
We can pay a query cost to learn the actual value, and we may allow an error threshold in the sorting.
The goal is to find a nearly-sorted permutation by performing a minimum-cost set of queries.

We show that an offline optimum query set can be found in polynomial time, and that both oblivious and adaptive problems have simple query-competitive algorithms.
The query-competitiveness for the oblivious problem is $n$ for uniform query costs, and unbounded for arbitrary costs; for the adaptive problem, the ratio is $2$.

We then present a unified adaptive strategy for uniform query costs that yields the following improved results:
(i) a $3/2$-query-competitive randomized algorithm;
(ii) a $5/3$-query-competitive deterministic algorithm if the dependency graph has no $2$-components after some preprocessing, which has query-competitive ratio $3/2 + \mathrm{O}(1/k)$ if the components obtained have size at least $k$;
and (iii) an exact algorithm if the intervals constitute a laminar family.
The first two results have matching lower bounds, and we have a lower bound of $7/5$ for large components.

We also give a randomized adaptive algorithm with query-competitive factor $1 + \frac{4}{3\sqrt{3}} \approx 1.7698$ for arbitrary query costs, and we show that the $2$-query competitive deterministic adaptive algorithm can be generalized for queries returning intervals and for a more general graph problem (which is also a generalization of the vertex cover problem), by using the local ratio technique.
Furthermore, we prove that the advice complexity of the adaptive problem is $\lfloor n/2 \rfloor$ if no error threshold is allowed, and $\lceil n/3 \cdot \lg 3 \rceil$ for the general case.

Finally, we present some graph-theoretical results regarding co-threshold tolerance graphs, and we discuss uncertainty variants of some classical interval problems.
\end{abstract}


\end{frontmatter}

\section{Introduction}

Sorting is one of the most fundamental problems in computer science and an essential part of any system dealing with large amounts of data.
High-performance algorithms such as QuickSort \cite{hoare62quicksort} have been known for decades, but the demand for fast sorting of huge amounts of data is such that improvements in sorting algorithms are still an active area of research; see, e.g., \cite{salah15sorting}.

In a distributed application with dynamic data, it may not be feasible to maintain a precise copy of the information in each replica.
In particular, accessing a local cached information may be much cheaper, even though not as precise, than querying a master database or to run a distributed consensus algorithm.
One approach is to maintain in the replicas, for each data item, an interval that bounds the actual value.
These intervals can be updated much faster than to guarantee a strict consistency of the data.
When higher precision is required, the system can query the master database for a more fine-grained interval or for the actual data value.
Therefore a trade-off between data precision and system performance can be established.
The TRAPP system, proposed by Olston and Widom \cite{olston2000queries}, relies on this concept.

This idea has led to theoretical investigation on \textbf{uncertainty problems with  queries}~\cite{bruce05uncertainty,erlebach16cheapestset,erlebach08steiner_uncertainty,feder07pathsqueires,feder03medianqueries,goerigk15knapsackqueries,megow17mst}.
Such problems also appear in optimization scenarios in which an extra effort can be incurred in order to obtain more precise values of the input data, such as by investing in market research, which is expensive so its cost should be minimized.
These works build upon more established frameworks of optimization with uncertainty, such as online \cite{borodin98online_alg}, robust~\cite{beyer07robsurvey} and stochastic~\cite{birge11stochastic} optimization.
In particular, the analysis of algorithms in terms of competitiveness against an adversary is inherited from the online optimization literature.

In this paper, we investigate the problem of sorting data items whose actual values are unknown, but for which we are given intervals on which the actual values lie.
We can {\bf query} an interval and then learn the actual value of the corresponding data item, but this incurs some cost.
The goal, then, is to sort the items by performing a set of queries of minimum cost.
Furthermore, the precision in the sorting may be relaxed, so that inversions may occur if the actual values are not too far apart.

\vspace{\baselineskip}

We distinguish between two types of algorithms for uncertainty problems with queries.
An {\bf adaptive} algorithm may decide which queries to perform based on results from previous queries.
An {\bf oblivious} algorithm, however, must choose the whole set of queries to perform in advance; i.e., it must choose a set of queries that certainly allow the problem to be solved without any knowledge of the actual values.
In this paper, both algorithms are compared with an offline optimum query set, i.e., a minimum-cost set of queries that proves the obtained solution to be correct.\footnote{This nomenclature differs to that used by Feder {\em et al.} \cite{feder03medianqueries}. They call an adaptive algorithm an {\bf online} algorithm, and an oblivious algorithm an {\bf offline} algorithm. We disagree with this nomenclature, since both types of algorithms are online in the standard sense of not knowing the data. Also, they compare an oblivious algorithm to an optimal oblivious strategy, and not to an offline optimum query set.}
An algorithm (either adaptive or oblivious) is {\bf $\alpha$-query-competitive} if it performs a total query cost of at most $\alpha$ times the cost of an offline optimum query set.

Another related problem is that of finding an {\bf optimum query set}.
Here we are given the actual data values, and want to identify a minimum-cost set of queries that would be sufficient to prove that the solution is correct.
Solving this problem is useful, for example, to perform experimental evaluation of online algorithms, since we are actually finding the offline optimum solution for the uncertainty problem.
This is also called the {\bf verification version} of the corresponding uncertainty problem with queries \cite{charalambous13uncertainty,erlebach16cheapestset}.

We are also interested in the {\bf advice complexity} of the problems we study.
In this setting, an online algorithm has access to an oracle that can give helpful information when making decisions.
The advice complexity is the number of bits of advice that are sufficient and necessary for an online algorithm to solve the problem exactly.
This is a research topic that has gained substantial attention; see~\cite{boyar17advice} for a survey.

\paragraph{\bf Our contribution}
We begin by showing how to compute an optimum query set in polynomial time, and that both oblivious and adaptive problems have simple algorithms with matching deterministic lower bounds.
The query-competitive ratio of the oblivious problem is $n$ if we have uniform query costs, and unbounded for arbitrary costs; for the adaptive problem, the query-competitive ratio is $2$.
The optimal oblivious algorithm is trivial; for the adaptive case, we have a simpler algorithm for uniform query costs, and a more sophisticated one for arbitrary query costs.
If query costs are uniform and the error threshold is zero, then the simpler algorithm can be implemented as an oracle for any comparison-based sorting algorithm, preserving time complexity and stability.

At this point it seems like the query-competitiveness of the problem is settled.
However, we present a unified adaptive strategy that attains different improvements for uniform query costs.
First, we obtain a 3/2-query-competitive algorithm by using randomization.
Second, if the error threshold is zero, and after some preprocessing the dependency graph has no $2$-components, the strategy yields a deterministic $5/3$-query-competitive algorithm; if the obtained graph has components of size at least $k$, then the same algorithm has query-competitive ratio $3/2 + \Oh(1/k)$.
The first two results have a matching lower bound, and for large components we have a lower bound of $7/5$.
The problem can also be solved exactly if the intervals constitute a laminar family.

We then present a randomized adaptive algorithm for arbitrary query costs.
We start with an algorithm with query-competitive factor $57/32 = 1.78125$, and then we improve it to $1 + \frac{4}{3\sqrt{3}} \approx 1.7698$ simply by changing the probabilities in the randomized step.
We also show that the $2$-query competitive deterministic algorithm for arbitrary query costs can be generalized for queries returning intervals, a model which was proposed in~\cite{gupta16queryselection}, and that this is the best possible factor for this case, even for a randomized algorithm.
The algorithm is adapted from the local ratio approximation algorithm for the vertex cover problem~\cite{baryehuda81vc}, and it can also be used to solve a more general graph problem which is also a generalization of the vertex cover problem.
Furthermore, we show that the advice complexity for adaptive algorithms is exactly $\lfloor n/2 \rfloor$ bits if there is no error threshold, and exactly $\lceil n/3 \cdot \lg 3 \rceil$ bits for the general case.

Finally, we present some results regarding the graph class defined by our sorting problem with uncertainty, which turns out to be the class of co-threshold tolerance (co-TT) graphs \cite{monma88ttolerance}.
We also discuss uncertainty variants of two classical interval problems, which inspired us to approach the sorting problem with uncertainty: the maximum independent set problem and the stabbing number problem.
Both problems have query-competitive factor at least $n-1$, where $n$ is the number of intervals, even if query costs are uniform and lower bounds are trivial.

\paragraph{\bf Related work}
The first work to investigate the minimum number of queries to solve a problem is by Kahan~\cite{kahan91queries}, who showed optimal adaptive strategies to find the minimum/maximum and median of $n$ values in uncertainty intervals, and for the sorting and closest pair problems.

Olston and Widom \cite{olston2000queries} proposed the TRAPP system, a distributed database based on uncertainty intervals.
The authors: (1) gave an optimal oblivious strategy for finding the minimum (and equivalently, the maximum) of a sequence of values within an error bound;
(2) showed that it is NP-hard to find an optimum oblivious query set to compute the sum of a sequence of values within an error bound, with a reduction from the knapsack problem.
The paper also discusses strategies for counting and finding the average of a sequence of values.
Khanna and Tan \cite{khanna01queries} generalized these results for arbitrary query costs and different levels of precision.

Feder {\em et al.} \cite{feder03medianqueries} considered the uncertainty version of the problem of finding the $k$-th largest value on a sequence (i.e., the generalized median problem).
The authors presented optimal oblivious and adaptive strategies for the problem, both running in polynomial time.
Both strategies are optimal, and the ratio between the oblivious and the adaptive strategy (also called the {\bf price of obliviousness}) is $\frac{2k-1}{k} < 2$ for uniform query costs, and $k$ for arbitrary query costs.\footnote{The works cited up to this point do not evaluate the algorithms using the competitiveness framework.}

Bruce {\em et al.} \cite{bruce05uncertainty} studied geometric problems where the points are given in uncertainty areas.
The authors gave 3-query-competitive algorithms for finding the maximal points and the convex hull in a two-dimensional space.
They also proposed the concept of {\bf witness sets}, which has been used subsequently in various works on uncertainty problems with queries.
Charalambous and Hoffman~\cite{charalambous13uncertainty} showed that it is NP-hard to find an optimum query set for the maximal points problem.

Feder {\em et al.} \cite{feder07pathsqueires} studied the uncertainty variant of the shortest path problem.
They showed that optimally solving the oblivious version of the problem is neither NP nor co-NP, unless NP = co-NP.
Their paper also discusses the complexity of the problem for various particular cases.

Erlebach {\em et al.} \cite{erlebach08steiner_uncertainty} proved that the minimum spanning tree problem with uncertainty admits an adaptive 2-query-competitive algorithm, which is the best possible for a deterministic algorithm.
Erlebach and Hoffman~\cite{erlebach14mstverification} showed that an optimum query set for the minimum spanning tree problem can be computed in polynomial time.
Erlebach, Hoffmann and Kammer~\cite{erlebach16cheapestset} studied a generalization called the {\bf cheapest set problem}, for which there is an adaptive algorithm with at most $d \cdot \opt + d$ queries, where $d$ is the maximum cardinality of a set.
They also generalized the result in~\cite{erlebach08steiner_uncertainty} to obtain an adaptive 2-query-competitive algorithm for the problem of finding a minimum-weight base on a matroid.

Gupta, Sabharwal and Sen \cite{gupta16queryselection} studied various of the previous problems in the setting where a query may return a refined interval, instead of the exact value of the data item.

Megow, Mei{\ss}ner and Skutella \cite{megow17mst} improved the result for the minimum spanning tree problem with a randomized adaptive algorithm, obtaining query-competitive ratio about~1.7.
(The problem has lower bound 1.5 for randomized algorithms.)
They also considered non-uniform query costs and proved that their results can be extended to find a minimum-weight base on a matroid.
Furthermore, they showed that the actual value of the minimum spanning tree can be computed in polynomial time.
Some experimental evaluation of those algorithms were presented in \cite{focke20mstexp}.

Ryzhov and Powell \cite{ryzhov12lpqueries} investigated how to solve a linear program while minimizing the query cost when the coefficients of the objective function are uncertain.
They presented a policy which is asymptotically optimal.
Maehara and Yamaguchi \cite{yamaguchi20ipqueries} studied the variant with packing constraints and coefficients following a probability distribution, and showed how to apply this to stochastic problems such as matching, matroid and stable set problems.

Note that all the work cited so far deals with problems whose classical (offline) versions can be solved in polynomial time.
Uncertainty versions with queries have been proposed for the knapsack problem \cite{goerigk15knapsackqueries} and the scheduling problem~\cite{arantes18schedulingqueries, durr2020scheduling}.
Since those problems are NP-hard, we might include the query cost into the solution cost and look for a competitive algorithm if we are looking for a polynomial-time algorithm.
Another option is to limit the maximum number of queries performed, and then to try to optimize the solution cost.

For a survey on the topic, see \cite{erlebach15querysurvey}.
Other references for related problems are also cited in~\cite{erlebach16cheapestset}.

Another sorting problem with uncertainty was studied by Ajtain {\em et al.}~\cite{ajtai16sortingnoise}.
In that problem, the values to be sorted are unknown, but their relative order can be tested by a comparison procedure.
However, comparing values that are too close returns imprecise answers, so in principle we should compare all ${n}\choose{2}$ pairs to obtain a sorting with some error guarantee.
The authors show how to solve the problem using only $\Oh(n^{3/2})$ comparisons.

\paragraph{\bf Organization of the paper}
In Section~\ref{sec:sorting}, we present the sorting problem with uncertainty and some basic facts, and in Section~\ref{sec:optoff} we give algorithms to find an offline optimum query set and for the oblivious setting.
We treat deterministic adaptive algorithms in Section~\ref{sec:adaptive}.
In Section~\ref{sec:rand}, we show how to improve the adaptive result for uniform query costs by using a randomized algorithm, or by assuming some structure in the dependency graph.
We present a randomized adaptive algorithm for arbitrary query costs in Section~\ref{sec:randcosts}, we discuss the variant of the problem in which queries may return intervals in Section~\ref{sec:cpcp}, and in Section~\ref{sec:advice} we investigate the advice complexity for adaptive algorithms.
Finally, in Section~\ref{sec:charac} we present some graph-theoretical results, and we discuss uncertainty variants of two classical interval problems in Section~\ref{sec:interval}.

\section{Sorting with Uncertainty}
\label{sec:sorting}

In the sorting problem with uncertainty, there are $n$ numbers $v_1, \ldots, v_n \in \reais$ whose exact value is unknown.
We are given $n$ uncertainty intervals $I_1, \ldots, I_n$ with $v_i \in I_i = [\ell_i, r_i]$, a cost $w_i \in \reaisnneg$ for querying interval $I_i$, and an error threshold $\delta \geq 0$.
After querying $I_i$, we obtain the exact value of $v_i$; we can also say that we replace $I_i$ with interval $I'_i = [v_i, v_i]$.
The goal is to obtain a permutation $\pi : [n] \rightarrow [n]$ such that $v_i \leq v_j + \delta$ if $\pi(i) < \pi(j)$ by performing a minimum-cost set of queries.

We begin by defining the following dependency relation between intervals, which is essential to solve the problem.

\begin{definition}
\label{def:dep}
Two intervals~$I_i$ and $I_j$ such that  $r_i - \ell_j > \delta$ and $r_j - \ell_i > \delta$ are {\bf dependent}.
Two intervals that are not dependent are {\bf independent}.
\end{definition}

\pagebreak

\begin{lemma}
\label{lemma:decideind}
The relative order between two intervals can be decided without querying either of them if and only if they are independent.
\end{lemma}

\begin{proof}
Let $I_i$ and $I_j$ be such that $r_i - \ell_j \leq \delta$.
Since $v_i \leq r_i$ and $v_j \geq \ell_j$, we have that $v_i \leq v_j + \delta$ and we can set $\pi(i) < \pi(j)$ without querying either of $I_i$ and $I_j$.

Conversely, let $I_j$ and $I_j$ be two dependent intervals.
We cannot set $\pi(i) < \pi(j)$, because it may be the case that $v_i = r_i$ and $v_j = \ell_j$, thus $r_i - \ell_j > \delta$ implies that $v_i > v_j + \delta$.
By a symmetric argument, we cannot set $\pi(j) < \pi(i)$, so we cannot decide the relative order between the intervals without making a query.
\qed
\end{proof}

So, essentially, to solve the sorting problem with uncertainty consists in querying intervals so that the graph defined by this dependency relation has no edges.
Querying both endpoints is sufficient to remove an edge, but sometimes it is enough to query one of them.
A solution is, therefore, a vertex cover in the dependency graph, but an optimum offline solution may not be a minimum vertex cover.

The graphs defined by this dependency relation are exactly the {\bf co-threshold tolerance} (co-TT) graphs~\cite{monma88ttolerance}.
$G = (V, E)$ is a {\bf threshold tolerance} graph if there are functions $w : V \rightarrow \reais$ and $t : V \rightarrow \reais$ such that $uv \in E$ if and only if $w(u) + w(v) \geq \min(t(u), t(v))$.
A co-TT graph is the complement of any threshold tolerance graph, or equivalently, $G = (V, E)$ is a co-TT graph if and only if there are functions $a : V \rightarrow \reais$ and $b : V \rightarrow \reais$ such that $uv \in E$ if and only if $a(u) < b(v)$ and $a(v) < b(u)$ \cite{monma88ttolerance}.

\begin{theorem}
\label{teo:cott}
The graphs defined by the dependency relation in Definition~\ref{def:dep} are exactly the co-TT graphs.
\end{theorem}

\begin{proof}
Given an instance $I_1, \ldots, I_n$ and an error threshold $\delta$, we simply define $a(I_i) := \ell_i$ and $b(I_i) := r_i - \delta$.

Conversely, given a co-TT graph $G = (V, E)$ with functions $a$ and $b$, we set
$\delta := \max_{v \in V} \{ a(v) - b(v), 0 \}$,
and for every $v \in V$ we set $\ell_v := a(v)$ and $r_v := b(v) + \delta$.
Both graphs have the same adjacency relation, $\delta \geq 0$, and $r_v \geq \ell_v$, so $[\ell_v, r_v]$ is an interval.
\qed
\end{proof}

The following result will be useful.

\begin{lemma}[\cite{monma88ttolerance}]
Every co-TT graph is chordal.
\end{lemma}

When $\delta > 0$, it is useful to distinguish intervals of width smaller than $\delta$, which we call {\bf trivial} intervals.
It is easy to check that two trivial intervals cannot be dependent, so when a trivial and a non-trivial interval are dependent, it is enough to query the non-trivial interval in order to decide their relative order.
This does not mean, however, that trivial intervals should never be queried, and in particular adaptive algorithms may decide to do that.

It is also clear that the dependency graph is an interval graph when $\delta = 0$.
This is also true when $\delta > 0$ and there are no trivial intervals, in which case we can simply replace each interval $I_i = [\ell_i, r_i]$ with $I^{(\delta)}_i = [\ell^{(\delta)}_i, r^{(\delta)}_i] := [\ell_i + \delta/2, r_i - \delta/2] \neq \emptyset$, and it is easy to check that~$I_i$ and $I_j$ are dependent with error threshold $\delta$ if and only if $I^{(\delta)}_i$ and $I^{(\delta)}_j$ are dependent with error threshold $0$.
Note however that we cannot use this reduction to solve the sorting problem, since the precise values could fall outside of the given interval.

\section{Warm-Up: Offline and Oblivious Algorithms}
\label{sec:optoff}

The first result we present concerns finding the optimum query set for a given set of intervals, assuming we know the actual values in each interval.
I.e., given the intervals $I_1, \ldots, I_n$ and the actual values $v_1, \ldots, v_n$, find a minimum-cost set $Q$ of intervals to query, such that $Q$ is sufficient to prove an ordering for $I_1, \ldots, I_n$ without the knowledge of $v_1, \ldots, v_n$.
Solving this problem is useful, for example, to perform experimental evaluation of algorithms, since we are actually finding the offline optimum solution for the online (either oblivious or adaptive) problem.
The ideas we present here will also be useful when solving the online problem.

We show that the problem can be solved optimally in polynomial time.
The key observations behind the algorithm are the following.
In order to simplify notation, we write $I_i \supset I_j$ for intervals $I_i$ and $I_j$ if $\ell_i < \ell_j$ and $r_j < r_i$.

\begin{fact}
\label{fact:offnecessary}
Let $I_i$ and $I_j$ be intervals with actual values $v_i$ and $v_j$. 
If $I_j \supset [v_i - \delta, v_i + \delta]$, then~$I_j$ is queried by every optimum solution.
\end{fact}

\begin{proof}
Even if we have queried $I_i$, we have to query $I_j$ because we may have $v_j \in [\ell_j, v_i - \delta)$ or $v_j \in (v_i + \delta, r_j]$.
\qed
\end{proof}

\begin{fact}
\label{fact:offdependent}
Let $I_i$ and $I_j$ be two dependent intervals, $v_i$ the actual value in~$I_i$ and~$v_j$ the actual value in $I_j$.
If $I_i \not\supset [v_j - \delta, v_j + \delta]$ and $I_j \not\supset [v_i - \delta, v_i + \delta]$, then it is enough to query either~$I_i$ or $I_j$ to decide their relative order.
\end{fact}

\begin{proof}
If we query $I_i$, then $v_i \notin [\ell_j + \delta, r_j - \delta]$, so we can pick a reasonable order between $I_i$ and $I_j$.
The argument is symmetrical if we query $I_j$.
\qed
\end{proof}

The algorithm begins with a query set $Q$ containing all intervals that satisfy the condition in Fact~\ref{fact:offnecessary}.
Due to Fact~\ref{fact:offdependent}, it is enough to complement $Q$ with a minimum-cost vertex cover in the dependency graph defined by the remaining intervals, which can be found in polynomial time for chordal graphs~\cite{gavril72chordal}.

\begin{theorem}
The problem of finding an optimum query set for the sorting problem with uncertainty can be solved optimally in polynomial time.
\end{theorem}

Now we consider oblivious algorithms.
In this case, all non-trivial intervals with some dependence must be queried, and clearly this is the best possible strategy.
In the following theorem, we show that this implies a tight bound of $n$ on the query-competitive ratio for the case with uniform costs, and that in the general case the query-competitive ratio is unbounded.

\begin{theorem}
If query costs are uniform, any oblivious algorithm for sorting with uncertainty has query-competitive ratio exactly $n$.
For arbitrary costs, the query-competitive ratio is unbounded.
\end{theorem}

\begin{proof}
For the upper bound with uniform costs, a naïve algorithm that queries all intervals and then sorts the numbers suffices.

For both lower bounds, we have $n-1$ independent intervals with length greater than $2\delta$, plus an interval~$I_n$ which contains all the other ones.
Both an algorithm and the optimum solution must query~$I_n$ in order to decide where~$v_n$ fits in the order.
If the algorithm does not query some $I_i$ with $i < n$, then the adversary can set $v_n \in (\ell_i + \delta, r_i - \delta) \neq \emptyset$ and the algorithm cannot decide the order.
Thus, without the knowledge of $v_n$, the algorithm must query all~$I_i$ with $i < n$.
However, it may be the case that $v_n \notin I_i$ for all $i < n$, and querying~$I_n$ suffices to decide the order.
This gives a lower bound of $n$ on the query-competitive ratio for uniform query costs.
For the general case, $w_n$ can be arbitrarily small and the query-competitive ratio is unbounded.
\qed
\end{proof}

\section{Deterministic Adaptive Algorithms}
\label{sec:adaptive}

Now let us consider deterministic adaptive algorithms.
We begin with a lower bound.

\begin{lemma}
\label{lemma:loweradap}
Any deterministic adaptive algorithm for the sorting problem with uncertainty has query-compet-itive ratio at least~$2$, even if query costs are uniform and the dependency graph has large components.
\end{lemma}

\begin{proof}
Consider intervals $I_1$ and $I_2$ with uniform query cost, $\ell_1 < \ell_2 < r_1 < r_2$ and $r_1 - \ell_2 > 2\delta$.
If the algorithm queries $I_1$, then the adversary chooses $v_1 \in (\ell_2 + \delta, r_1 - \delta)$.
The algorithm must also query~$I_2$ to decide the order, but then the adversary can choose $v_2 \in [r_1 - \delta, r_2]$ and one query would be sufficient.
The argument is symmetrical if the algorithm queries $I_2$ first, with $v_2 \in (\ell_2 + \delta, r_1 - \delta)$ and $v_1 \in [\ell_1, \ell_2 + \delta]$.
To obtain a large component, make several independent copies of this structure and connect them with a large interval containing all the others; in this case the lower bound approaches 2 asymptotically.
\qed
\end{proof}

First we give a simple deterministic 2-query-competitive adaptive algorithm for the case with uniform query costs.
It is inspired by the algorithm of Erlebach~{\em et~al.}~\cite{erlebach08steiner_uncertainty} for the minimum spanning tree problem with uncertainty, and it relies on the following concepts, which were introduced in~\cite{bruce05uncertainty}.
Let $\Ical = \{I_1, \ldots, I_n\}$ be a set of intervals for the sorting problem with uncertainty.
We say that a set $W \subseteq \Ical$ of intervals is a {\bf witness set} if at least one of the intervals in $W$ must be queried to decide the order of~$\Ical$, even if all intervals except those in~$W$ are queried.
Due to Lemma~\ref{lemma:decideind}, any pair of dependent intervals constitute a witness set.
A set of intervals $\Ical' = \{I'_1, \ldots, I'_n\}$ is a {\bf refinement} of $\Ical$ if $\Ical'$ is obtained from~$\Ical$ by performing a sequence of queries.
Fact~\ref{fact:refinement} follows simply from $\Ical'$ having more information than $\Ical$.

\begin{fact}
\label{fact:refinement}
Let $\Ical'$ be a refinement of $\Ical$.
If some set of intervals $W \subseteq \Ical' \cap \Ical$ is a witness set for $\Ical'$, then it is a witness set for $\Ical$.
\end{fact}

The algorithm, then, consists in the following.
While there is some pair of dependent intervals, we query all intervals in this pair that have not been queried yet.
When an interval~$I_i$ is queried, it is replaced by $[v_i, v_i]$.
(Note that, even after querying $I_i$, it may still be dependent to a non-trivial interval.)
Finally, intervals are sorted by breaking ties arbitrarily.\footnote{If $\delta = 0$, then this algorithm can be implemented with stable sorting and in $\Oh(n \lg n)$ time by running a standard stable sorting algorithm (e.g., MergeSort) and querying two intervals when MergeSort needs to know the relative order between them.
It does not work, however, if $\delta > 0$, since the relation $v_i \leq v_j + \delta$ is not transitive.}

For a better understanding of the algorithm, consider the examples in Figure~\ref{fig:interval}, assuming $\delta = 0$.
In Figure~\ref{fig:interval1}, the optimum solution must query $I_1$ and $I_3$, since $v_1 \in I_3$ and $v_3 \in I_1$, and this is enough because~$I_2$ will be independent after querying~$I_1$.
If the algorithm first queries~$I_1$ and $I_2$, it must also query~$I_3$.
In Figure~\ref{fig:interval3} it is enough to query $I_1$, but the algorithm will query a dependent pair, say,~$I_1$ and~$I_2$.
Either way, the algorithm does not spend more than twice the optimum number of queries.
The following theorem was previously proven for $\delta = 0$ by Kahan~\cite{kahan91queries}, which we generalize for arbitrary $\delta \geq 0$.

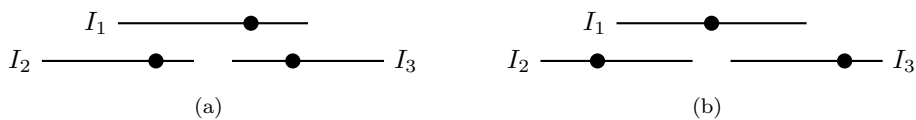
\begin{figure}[!ht]
  \centering
  \subfigure[]{\label{fig:interval1}
   \begin{tikzpicture}[thick, scale=0.5]
    \draw (0, 0) node[anchor=east]{$I_2$} -- (4, 0);
    \draw (2, 1) node[anchor=east]{$I_1$} -- (7, 1);
    \draw (5, 0) -- (9, 0) node[anchor=west] {$I_3$};
    \fill[black] (3, 0) circle (0.2cm);
    \fill[black] (5.5, 1) circle (0.2cm);
    \fill[black] (6.6, 0) circle (0.2cm);
   \end{tikzpicture}
  }\qquad
  \subfigure[]{\label{fig:interval3}
   \begin{tikzpicture}[thick, scale=0.5]
    \draw (0, 0) node[anchor=east]{$I_2$} -- (4, 0);
    \draw (2, 1) node[anchor=east]{$I_1$} -- (7, 1);
    \draw (5, 0) -- (9, 0) node[anchor=west] {$I_3$};
    \fill[black] (1.5, 0) circle (0.2cm);
    \fill[black] (4.5, 1) circle (0.2cm);
    \fill[black] (8, 0) circle (0.2cm);
   \end{tikzpicture}
  }
  \caption{Example instances of the problem.}
  \label{fig:interval}
\end{figure}

\begin{theorem}
\label{teo:simple}
The simple adaptive algorithm for sorting with uncertainty is 2-query-competitive for uniform query costs.
\end{theorem}

\begin{proof}
Note that the optimum solution must query at least one interval in each witness set.
For every pair $\{I_i, I_j\}$ of dependent intervals selected by the algorithm, we have that: (1)~if both $I_i$ and $I_j$ have not been queried yet, the algorithm queries the witness set $\{I_i, I_j\}$; (2)~if~$I_i$ has already been queried then, by Fact~\ref{fact:offnecessary}, $\{I_j\}$ is a witness set, which is queried by the algorithm.
We conclude that the algorithm only queries disjoint witness sets of size at most 2, thus it queries at most twice the minimum number of intervals.
\qed
\end{proof}

For arbitrary query costs, the problem also admits a 2-query-competitive deterministic adaptive algorithm, although not as simple.
The algorithm first queries a minimum-cost vertex cover $S_1$ on the dependency graph.
Then, it queries all non-trivial intervals that are still dependent after querying~$S_1$, which we denote by the set $S_2$.

\begin{theorem}
\label{teo:detarbitrary}
The adaptive algorithm for sorting with uncertainty with arbitrary query costs is 2-query-competitive.
\end{theorem}

\begin{proof}
Let $Q$ be an optimum query set.
The set of intervals not contained in $Q$ must be independent.
By the duality between independent sets and vertex covers, $Q$ must be a vertex cover.
Thus $w(S_1) \leq w(Q)$, since $S_1$ has minimum cost.
Furthermore, note that every interval in $S_2$ is a singleton witness set, since~$S_2$ is a set of independent intervals.
Thus $w(S_2) \leq w(Q)$ as well, and $w(S_1 \cup S_2) \leq 2 \cdot w(Q)$.
\qed
\end{proof}

\section{Improved Adaptive Algorithms for Uniform Query Costs}
\label{sec:rand}

We now explore refined analysis of query-competitive sorting. We present a unified strategy that yields different improvements to Theorem~\ref{teo:simple}, depending on what assumptions we make.

The core observation is that the bad 2-interval instance in the proof of Lemma~\ref{lemma:loweradap} is the only structure that prevents an algorithm from performing better than twice the optimum.
The first strategy that comes to mind, then, is to use randomization: a simple randomized strategy attains query-competitive factor~$3/2$ on the instance of Lemma~\ref{lemma:loweradap}.
Before extending the algorithm to arbitrary instances, we give a lower bound for any randomized algorithm.

\begin{lemma}
\label{lemma:lowerrand}
Any randomized adaptive algorithm has query-competitive ratio at least~$3/2$ against an adversary that is oblivious to the randomized tosses, even for uniform query costs.
\end{lemma}

\begin{proof}
Use the same bad instance as Lemma~\ref{lemma:loweradap}, set probability $1/2$ for each of the two possible inputs and apply Yao's minimax principle.
\qed
\end{proof}

The algorithm is based on the following property of the dependency graph.

\begin{lemma}
\label{lemma:simplicial}
If $I_x$ is a dependent interval with minimum $r_x$, then vertex $x$ is simplicial, i.e., its neighborhood is a clique.
\end{lemma}

\begin{proof}
The claim is trivial if $x$ has only one neighbor, so assume it has at least two, $y$ and $z$.
Then $r_y - \ell_z \geq r_x - \ell_z > \delta$, since $I_x$ and $I_z$ are dependent.
Analogously, $r_z - \ell_y > \delta$, so $I_y$ and $I_z$ are dependent.
\qed
\end{proof}

The algorithm begins by querying intervals that are singleton witness sets according to a generalization of the condition in Fact~\ref{fact:offnecessary}.
Then, if a component of the remaining dependency graph is an edge, the randomized strategy is applied.
Else, the algorithm considers a non-isolated vertex $x$ with minimum $r_x$, a neighbor $y$ of~$x$ with minimum $r_y$, and another neighbor $z$ of $x$ (or of $y$ if $y$ is the only neighbor of $x$) with minimum~$r_z$. The algorithm first queries $I_y$.
If $x$ and $y$ are still adjacent, or if $x$ and $z$ are adjacent, then we query both~$I_x$ and $I_z$.
We repeat this strategy until the dependency graph has no edges.

A pseudocode is presented in Algorithm~\ref{alg:rand}; we parameterize the probability $p$ in the randomized strategy since the algorithm will be reused afterwards.
We also maintain a set $\Vcal$ of the values resulting of queried intervals.

\begin{algorithm}[!ht]
\KwIn{$(I_1, \ldots, I_n, p)$}
\SetAlgoNoEnd
$\Vcal \recebe \emptyset$\;
\While{there are $i, j$ with $I_i \supset [\ell_j - \delta, r_j + \delta]$ \KwOr $I_i \supset [v_j - \delta, v_j + \delta]$ with $v_j \in \Vcal$}{
 query $I_i$, add $v_i$ to $\Vcal$\;
}
\While{there is some dependency}{
 \eIf{some component is an edge $ij$}{
  pick $i$ with probability $p$ (and $j$ with probability $1-p$); assume $i$ is picked\;
  query $I_i$, add $v_i$ to $\Vcal$\;
  \If{$I_j \supset [v_i - \delta, v_i + \delta]$}{
   query $I_j$, add $v_j$ to $\Vcal$\;
  }
 }{
  \KwLet $I_x$ non-isolated with $\min r_x$, and $y$ be a neighbor of $x$ with $\min r_y$\;
  \KwLet $z$ be another neighbor of $x$ (or of $y$ if $x$ is a leaf), with $\min r_z$\;
  query $I_y$, add $v_y$ to $\Vcal$\;
  \If{$I_x \supset [v_y - \delta, v_y + \delta]$ \KwOr $I_x, I_z$ are dependent}{
   query $I_x$, add $v_x$ to $\Vcal$\;
   query $I_z$, add $v_z$ to $\Vcal$\;
  }
 }
 \While{there is $I_i \supset [v_j - \delta, v_j + \delta]$ for some $v_j \in \Vcal$}{
  query $I_i$, add $v_i$ to $\Vcal$\;
 }
}
\vspace{0.2cm}
\caption{\label{alg:rand} Improved adaptive algorithm for the sorting problem with queries.}
\end{algorithm}

\begin{theorem}
\label{teo:rand}
Algorithm~\ref{alg:rand} has expected query-competitive ratio $3/2$ if $p = 1/2$.
\end{theorem}

\begin{proof}
We form a partition $V_1, \ldots, V_m$ of the set of input intervals with the following property.
Let $a(V_i)$ be the number of intervals in $V_i$ that are queried by the algorithm, and let $q(V_i) := |Q \cap V_i|$, where $Q$ is an optimum query set.
We show that $\Esp[a(V_i) / q(V_i)] \leq 3/2$ for every $i$, from which the theorem follows.

If the algorithm queries an interval $I_i$ in Line~3 or Line~18, then $\{I_i\}$ is the next set in the partition.
Due to Fact~\ref{fact:offnecessary}, it is a singleton witness set, so $a(\{I_i\}) / q(\{I_i\}) = 1$.

If the algorithm runs Lines~6--9 for edge $ij$, then $W = \{I_i, I_j\}$ is the next set in the partition.
We consider the following cases.
\begin{enumerate}
 \item If $I_i \supset [v_j - \delta, v_j + \delta]$ and $I_j \supset [v_i - \delta, v_i + \delta]$, then $q(W) = 2$ and $a(W) = 2$.
 \item Otherwise, $q(W) \geq 1$ because this is a witness pair.
 \begin{enumerate}
  \item If $I_i \supset [v_j - \delta, v_j + \delta]$ but $I_j \not\supset [v_i - \delta, v_i + \delta]$, then with probability $1/2$ the algorithm queries~$I_i$ and this is enough, and with probability $1/2$ it queries both, so $\Esp[a(W)] = 3/2$; the same holds for the symmetrical case.
  \item If $I_i \not\supset [v_j - \delta, v_j + \delta]$ and $I_j \not\supset [v_i - \delta, v_i + \delta]$, then Line~9 is not executed and $a(W) = 1$.
 \end{enumerate}
\end{enumerate}

If the algorithm runs Lines~11--16 for $x$, $y$ and $z$, then we have two cases.
\begin{enumerate}
 \item If $x$ and $z$ are not neighbors, and $x$ and $y$ are not neighbors after Line~13, then $W = \{I_x, I_y\}$ is the next set in the partition.
  Since it is a witness set, $q(W) \geq 1$.
  But the algorithm will not query $I_x$ because $y$ is the only neighbor of $x$, so $a(W) = 1$.
 \item Otherwise, $W = \{I_x, I_y, I_z\}$ is the next set in the partition.
 We have two subcases.
 \begin{enumerate}
  \item If $x$ and $z$ are neighbors, then $xyz$ is a clique by Lemma~\ref{lemma:simplicial}.
  So $q(W) \geq 2$, since otherwise a pair is unsolved.
  \item Otherwise, $I_x \supset [v_y - \delta, v_y + \delta]$ and $\{I_x\}$ is a singleton witness set.
  Since $x$ and $z$ are not neighbors, then $y$ and $z$ are neighbors and, by Lemma~\ref{lemma:decideind}, $\{I_y, I_z\}$ is a witness set.
 \end{enumerate}
 Either way, $q(W) \geq 2$ and $a(W) = 3$.
\end{enumerate}
We conclude that the expected query-competitive ratio is $3/2$.
\qed
\end{proof}

Our second strategy to obtain an improvement on Theorem~\ref{teo:simple} is, instead of using randomization, to assume that the graph does not have $2$-components, i.e., components consisting of a single edge.
This is not enough, however, since in Lemma~\ref{lemma:loweradap} we have shown that we can have a large component.
So our hypothesis is that $\delta = 0$ and, after executing the loop of Lines~2--3, the remaining dependency graph, which becomes a proper interval graph, has no $2$-components.
(Note that Theorem~\ref{teo:rand} is still true if we remove Lines~2--3 of the algorithm.)
Let us prove a lower bound for this case.

\begin{lemma}
\label{lemma:lower2comp}
Any deterministic adaptive algorithm has query-competitive ratio at least~$5/3$, even if $\delta = 0$ and the dependency graph is a proper interval graph with no $2$-components.
\end{lemma}

\begin{proof}
Consider five proper intervals $I_a, I_b, I_c, I_d, I_e$ with $\ell_a < \ell_b < \ell_c < \ell_d < \ell_e$.
The dependencies are defined by two triangles, $abc$ and $cde$.

If the algorithm first queries $I_c$, then we set $v_c \in I_c \setminus (I_a \cup I_b \cup I_d \cup I_e)$, and we can make $ab$ and $de$ behave as the bad instance of Lemma~\ref{lemma:loweradap}.

If the algorithm first queries $I_a$, then we set $v_a \in I_b \cap I_c$, so the algorithm will be forced to query $I_b$ and~$I_c$, and we set $v_b, v_c \in (I_b \cup I_c) \setminus (I_a \cup I_d \cup I_e)$, so the optimum can avoid $I_a$.
Then we can make $de$ behave as the bad instance of Lemma~\ref{lemma:loweradap}.
The argument is symmetric if the algorithm first queries $I_e$.

If the algorithm first queries $I_b$, then we set $v_b \in I_a \cap I_c$, so the algorithm will be forced to query $I_a$ and~$I_c$, and we set $v_a, v_c \in (I_a \cup I_c) \setminus (I_b \cup I_d \cup I_e)$, so the optimum can avoid $I_b$.
Then we can make $de$ behave as the bad instance of Lemma~\ref{lemma:loweradap}.
The argument is symmetric if the algorithm first queries $I_d$.
\qed
\end{proof}

\begin{theorem}
\label{teo:proper2comp}
Algorithm~\ref{alg:rand} (with $p=0$ or $1$) is $5/3$-query-competitive if $\delta = 0$ and the dependency graph has no $2$-components after finishing the loop of Lines~2--3.
\end{theorem}

\begin{proof}
The analysis is similar to that of Theorem~\ref{teo:rand}.
We will give a partition $V_1, \ldots, V_m$ of the set of intervals with the following property.
Let $a(V_i)$ be the number of intervals in $V_i$ that are queried by the algorithm, and let $q(V_i) := |Q \cap V_i|$, where $Q$ is an optimum query set.
We will have that $a(V_i) / q(V_i) \leq 5/3$ for every $i$, and then the theorem follows.
The analysis for the cases of Lines~3, 11--16 and~18 are identical.

If the algorithm runs Lines~6--9 for edge $ij$, then let $C$ be the component containing $ij$ in the dependency graph after finishing the loop of Lines~2--3.
We claim that $i$ and $j$ are the only vertices of $C$ queried in Lines~6--9: Lines~11--16 force that intervals are queried from left to right; thus, since the dependency graph at this point is a proper interval graph, if an interval $i'$ is queried in Line~18, then after that no interval $j'$ with $r_{j'} < r_{i'}$ will have some dependency.
Pick an arbitrary set~$W'$ of the partition consisting of vertices of~$C$.
We merge $\{I_i, I_j\}$ and $W'$ into a single set~$W$ of the partition, and from the previous cases we have that $a(W) / q(W) \leq 5/3$.
\qed
\end{proof}

This proof indicates that the analysis can be improved if we require the graph to have {\em large} components after finishing the loop of Lines~2--3.

\begin{theorem}
\label{teo:properlarge}
Algorithm~\ref{alg:rand} (with $p=0$ or $1$) has query-competitive factor $3/2 + \Oh(1/k)$ if $\delta = 0$ and each component of the dependency graph has size at least $k$ after finishing the loop of Lines~2--3.
\end{theorem}

\begin{proof}
We only have to reconsider the case of Lines~6--9 in the proof of Theorem~\ref{teo:proper2comp}.
If the algorithm runs Lines~6--9 for edge $ij$, then let $C$ be the component containing $ij$ in the dependency graph after finishing the loop of Lines~2--3.
We merge $\{I_i, I_j\}$ and all partition sets containing vertices of $C$ into a single set $W$ of the partition.
Since $i$ and $j$ are the only vertices of $C$ queried in Lines~6--9 and $C$ has size at least $k$, from the other cases we have that $a(W) / q(W) \leq 3/2 + \Oh(1/k)$.
\qed
\end{proof}

The analysis is tight since we can have a chain of $k$ triangles plus $1$ edge, such that we can force the algorithm to query all intervals, while the optimum can avoid one interval in each triangle and one interval in the extra edge.
For large components, we still have a lower bound of $7/5$ for any deterministic algorithm.

\begin{lemma}
\label{lemma:lowerproperlarge}
Any deterministic adaptive algorithm has query-competitive ratio at least~$7/5$, even if $\delta = 0$ and the dependency graph is a proper interval graph with large components.
\end{lemma}

\begin{proof}
Consider the graph of Figure~\ref{fig:lowerlarge}, which has $7k+2$ vertices.
For $i = 0, \ldots, k-1$, vertices $7i + 3, \ldots, 7i+7$ consist in a copy of the instance of Lemma~\ref{lemma:lower2comp}.
For $i = 0, \ldots, k$, vertices $x_i = 7i+1$ and $y_i = 7i+2$ are dependent, $x_i$ is dependent to $7(i-1)+7$ if $i > 0$, and $y_i$ is dependent to $7i+3$ if $i < k$.
We set $v_{x_i}, v_{y_i} \in I_{x_i} \cap I_{y_i}$, so both the algorithm and the optimum must query $I_{x_i}$ and $I_{y_i}$, but querying them gives us no information about the remaining vertices.
From Lemma~\ref{lemma:lower2comp}, we can force any deterministic algorithm to query all vertices in the graph, while the optimum solution can query only $3$ vertices of $7i+3, \ldots, 7i+7$.
\qed
\end{proof}

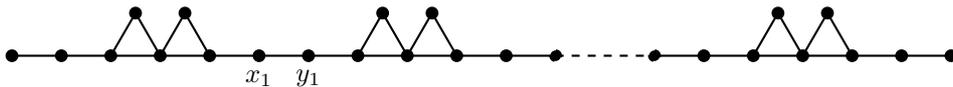
\begin{figure}[!ht]
  \centering
  \tikzstyle{every node}=[circle, draw, fill=black, inner sep=0pt, minimum width=4pt]
  \begin{tikzpicture}[thick, scale=0.65]
   \draw \foreach \x in {0, 1, ..., 10} {
     (\x, 0) node[label=south:\ifthenelse{\x=5}{$x_1$}{\ifthenelse{\x=6}{$y_1$}{}}]{} -- (\x + 1, 0)
   };
   \draw[dashed] (11, 0) node{} -- (13, 0) node{};
   \draw \foreach \x in {13, 14, ..., 18} {
     (\x, 0) -- (\x + 1, 0) node{}
   };
   \draw (2, 0) -- (2.5, 0.87) node{};
   \draw (2.5, 0.87) -- (3, 0);
   \draw (3, 0) -- (3.5, 0.87) node{};
   \draw (3.5, 0.87) -- (4, 0);
   \draw (7, 0) -- (7.5, 0.87) node{};
   \draw (7.5, 0.87) -- (8, 0);
   \draw (8, 0) -- (8.5, 0.87) node{};
   \draw (8.5, 0.87) -- (9, 0);
   \draw (15, 0) -- (15.5, 0.87) node{};
   \draw (15.5, 0.87) -- (16, 0);
   \draw (16, 0) -- (16.5, 0.87) node{};
   \draw (16.5, 0.87) -- (17, 0);
  \end{tikzpicture}
  \caption{Instance which attains the lower bound for proper interval graphs with large components.}
  \label{fig:lowerlarge}
\end{figure}

It remains an open question to close the gap between the lower bound of $7/5$ and the upper bound of $3/2 + \Oh(1/k)$.
Finally, we note that the problem can be solved exactly for laminar families of intervals\footnote{A set of intervals $\{I_1, \ldots, I_n\}$ is a laminar family if, for every $I_i, I_j$ with $I_i \cap I_j \neq \emptyset$, we have that either $I_i \subset I_j$ or $I_j \subset I_i$.}, since all queries will happen at Line~3 of the algorithm.

\begin{theorem}
Algorithm~\ref{alg:rand} obtains an optimum solution if $\delta = 0$ and the intervals constitute a laminar family.
\end{theorem}

\section{Randomized Adaptive Algorithms for Arbitrary Query Costs}
\label{sec:randcosts}

In this section, we improve over Theorem~\ref{teo:detarbitrary}.
We begin by showing that, when query costs are arbitrary, any deterministic algorithm has query-competitive ratio at least 2, even for proper interval graphs with large components.

\begin{lemma}
\label{lemma:lowerproperlargecosts}
Any deterministic adaptive algorithm for arbitrary query costs has query-competitive ratio at least~$2$, even if $\delta = 0$ and the dependency graph is a proper interval graph with large components.
\end{lemma}

\begin{proof}
Consider a path with $2n$ intervals.
The first two intervals have query cost 1, and consist in the bad 2-component instance of Lemma~\ref{lemma:loweradap}.
The other $2n-2$ intervals have query cost $0 < \epsilon \ll 1/(2n)$, and we force them all to be in any solution by setting $v_{2i+1} \in I_{2i+2}$ and $v_{2i+2} \in I_{2i+1}$, for $i = 1, \ldots, n-1$.
\qed
\end{proof}

We can obtain an improvement, however, by using randomization.
First, consider a 2-component with intervals $I_a$ and $I_b$.
A simple strategy is to query $I_b$ first with probability $p = \min \left( 1, \displaystyle\frac{w_a}{2w_b} \right)$ (and $I_a$ with probability $1-p$), and then query the other interval if needed.
If the optimum solution queries both $I_a$ and~$I_b$, then the algorithm is optimal.
If $w_a \geq 2w_b$ then the algorithm is deterministic but clearly pays at most $3/2$ times the optimum.
Else, if the optimum solution is to query only $I_a$, then the algorithm pays at most
\begin{displaymath}
 \left( 1 - \frac{w_a}{2w_b} \right) \cdot w_a + \frac{w_a}{2w_b} \cdot (w_b + w_a)
 = \frac{3}{2} w_a.
\end{displaymath}
If the optimum solution is to query only $I_b$, then the algorithm pays at most
\begin{eqnarray*}
 \frac{w_a}{2w_b} \cdot w_b + \left( 1 - \frac{w_a}{2w_b} \right) \cdot (w_a + w_b)
 & = & w_b \cdot \left( 1 + \frac{w_a}{w_b} \cdot \left( 1 - \frac{w_a}{2w_b} \right) \right) \\
 & = & w_b \cdot \left( 1 + \frac{w_a}{w_b} \cdot \frac{w_b^2 - (w_a - w_b)^2}{2 w_a w_b} \right)\\
 & \leq & w_b \cdot \left( 1 + \frac{w_a}{w_b} \cdot \frac{w_b}{2w_a} \right) = \frac{3}{2} \cdot w_b,
\end{eqnarray*}
where the inequality follows from the fact that $(w_a - w_b)^2 \geq 0$.
Note that this strategy is optimal due to Lemma~\ref{lemma:lowerrand}.
(The same result works if we pick $I_a$ with probability $p' = \min \left( 1, \displaystyle\frac{w_b}{2w_a} \right)$ and $I_b$ with probability $1 - p'$, but we will favor $I_b$ to simplify the discussion hereon.)

Now let us discuss how to obtain a general strategy for arbitrary graphs.
The first key ingredient is the Local Ratio Theorem, which we state below (see, e.g., \cite{baryehuda04localratio} for a proof).

\begin{theorem}[Local Ratio]
\label{teo:localratio}
Consider a minimization problem whose objective function is a linear combination of the solution vector.
Let $w, w^{(1)}, w^{(2)} \in \reais^n$ be cost vectors such that $w = w^{(1)} + w^{(2)}$.
Let $x \in \reais^n$ be a feasible solution, and let $x^*, x^{*(1)}, x^{*(2)} \in \reais^n$ be optimum solutions for costs $w, w^{(1)}, w^{(2)}$, respectively.
If, for some constant $\alpha > 0$, $w^{(1)}x \leq \alpha w^{(1)} x^{*(1)}$ and $w^{(2)}x \leq \alpha w^{(2)} x^{*(2)}$, then $wx \leq \alpha w x^*$.
\end{theorem}

We use this result to eliminate triangles in the dependency graph.
We begin by querying intervals with zero query cost, since they do not change the solution cost.
Then we consider an arbitrary triangle $abc$ and let $i$ denote the interval in $\{a,b,c\}$ of minimum query cost.
We set $w_j \recebe w_j - w_i$ for $j \in \{a, b, c\}$, which by the Local Ratio theorem makes progress towards a 3/2-approximation. 

After breaking all triangles, the dependency graph becomes a forest, since it is a chordal graph.
In fact, it becomes a caterpillar (where all vertices are at distance at most 1 from a longest path), since co-TT graphs cannot contain the graph of Figure~\ref{fig:asteroid1} as an induced subgraph~\cite{monma88ttolerance}.
We identify a longest path in a component (which can be done in polynomial time for trees) and let $a,b,c$ be its first three nodes.
We focus on the node $b$ and the set $N_1(b)$ of the neighbors of $b$ excluding $c$.
Since the graph is a caterpillar, note that $b$ is the only neighbor of the vertices in $N_1(b)$.
Let $w(N_1(b)) = \sum_{i \in N_1(b)} w_i$.
We query $I_b$ with probability $p = \min \left( 1, \displaystyle\frac{w(N_1(b))}{2w_b} \right)$, and we query all intervals in $N_1(b)$ with probability $1-p$.

We repeat this strategy and query singleton witness sets until the dependency graph has no edges.
A pseudocode is presented in Algorithm~\ref{alg:randcost}.
Note that, after all triangles are broken, the longest path of each component can be computed just once; we enforce this, so that we can assume that the longest path is consistent between executions of Lines~10--14 in the same component.

\begin{algorithm}[!ht]
\SetAlgoNoEnd
\KwIn{$(I_1, \ldots, I_n, w)$}
$\Vcal \recebe \emptyset$\;
\While{there is some dependency}{
 \uIf{some $I_i$ has $w_i = 0$}{
  query $I_i$, add $v_i$ to $\Vcal$\;  
 }
 \uElseIf{there is a triangle $abc$}{
  \KwLet $i = \arg \min_{i \in \{a, b, c\}} \{ w_i \}$\;
  \ForEach{$j \in \{a, b, c\}$}{
   $w_j \recebe w_j - w_i$\;
  }
 }
 \Else{
  \KwLet $P = abc\cdots$ be a longest path in a component\;
  \WithProb{$p = \min \left( 1, \displaystyle\frac{w(N_1(b))}{2w_b} \right)$}{
   query $I_b$, add $v_b$ to $\Vcal$\;
  }
  \ElseForEach{$j \in N_1(b)$}{
   query $I_j$, add $v_j$ to $\Vcal$\;
  }
 }
 \While{there is $I_i \supset [v_j - \delta, v_j + \delta]$ for some $v_j \in \Vcal$}{
  query $I_i$, add $v_i$ to $\Vcal$\;
 }
}
\vspace{0.2cm}
\caption{\label{alg:randcost} Randomized adaptive algorithm for arbitrary query costs.}
\end{algorithm}

\begin{theorem}
Algorithm~\ref{alg:randcost} has query-competitive ratio $57/32 = 1.78125$ for arbitrary query costs.
\end{theorem}

\begin{proof}
If an interval is queried at Line~4, then it does not change the solution cost.
Also, recall that, due to Fact~\ref{fact:offnecessary}, every interval queried in Line~16 is in the optimum solution.

If the dependency graph contains triangles, then we do an induction.
The base case is when all intervals have zero cost, which has a trivial optimum solution, or when the dependency graph is a forest, for which we will prove later that the algorithm has factor $57/32$.
If the algorithm considers a triangle $abc$ in Lines~6--8, then we apply the Local Ratio Theorem with
$w^{(1)}(j) = \left\{ \begin{array}{ll}
 w_i, & \mbox{if } j \in \{a, b, c\}\\
 0, & \mbox{otherwise}
\end{array} \right.$
and $w^{(2)} = w - w^{(1)}$, where~$w$ is the cost vector before Line~7.
By induction hypothesis, the solution returned by the algorithm has query-competitive factor $57/32$ for $w^{(2)}$.
The solution returned by the algorithm in the worst case queries all intervals in $abc$; since in any triangle at least two intervals must be queried, the solution costs at most $3/2$ times the optimum for $w^{(1)}$.
By the Local Ratio Theorem, the algorithm has query-competitive factor $57/32$ for $w$.
In the following, we focus on the case where each component of the dependency graph is a tree (specifically, a caterpillar).

Let us define some terminology.
A {\bf random trial} is an execution of Lines~11--14 of the algorithm.
We say that a random trial {\bf involves} interval $I_i$, and $I_i$ {\bf is involved} in the random trial, if $i \in N_1(b) \cup \{b\}$.
An interval $I_i$ {\bf is queried in} a random trial if it is queried in an execution of Line~12 or~14 which is part of that random trial.
Note that some intervals involved in a random trial may not be queried in the random trial, but may be queried later on, either in another random trial or in Line~16.

We bound the cost of the returned solution by assigning a cost share to each interval in the optimum solution.
Intervals that are queried and are in the optimum solution have their cost assigned to themselves.
Intervals that are queried and are not in the optimum solution have their cost assigned to their neighbors in the optimum solution.
However, we only assign a query cost $w_j$ to an interval $I_i$ with $i \neq j$ if $I_j$ is queried in a random trial involving $I_i$.
(Note that intervals that are not in the optimum solution can only be queried in random trials.)
Also, if $I_j$ has more than one neighbor, its cost is shared among its neighbors in proportion to their cost.

Let us bound the cost share for each interval in the optimum solution.
Note that we do not have to care about intervals that are not involved in random trials, since they are not assigned extra cost.
Let $P$ and $b$ be as defined in Line~10.
Note that, if the optimum solution does not query $I_b$, then it must query all neighbors of $b$.
If $w(N_1(b)) \geq 2w_b$, then the algorithm deterministically queries $I_b$ first, and spends at most $3/2$ times the portion of the optimum solution contained in $N_1(b) \cup \{b\}$.
If $w(N_1(b)) < 2w_b$, then we consider two cases.

\begin{enumerate}
 \item If the optimum solution does not query $I_b$, then it queries all of $N_1(b)$. The expected total cost share for intervals in $N_1(b)$ will be at most
\begin{displaymath}
 \left( 1 - \frac{w(N_1(b))}{2w_b} \right) \cdot w(N_1(b)) + \frac{w(N_1(b))}{2w_b} \cdot (w_b + w(N_1(b)))
 = \frac{3}{2} w(N_1(b)).
\end{displaymath}
Remember that we split this cost among the intervals in $N_1(b)$ in proportion to their cost, so the expected cost share for $j \in N_1(b)$ will be at most $\frac{3}{2} w_j$.
Note that we do not assign any cost share to~$I_b$, since it is not in the optimum solution.
Also, if $I_b$ is not queried in this trial but is queried later on, its cost will not be assigned to $N_1(b)$.

 \item If the optimum solution queries $I_b$, then the algorithm may not query $I_b$ in the trial between~$b$ and $N_1(b)$, but in this case $I_b$ will be involved in the trial between $c$ and the neighbors of $c$, and in no further trials.
All intervals queried in the first trial that are not in the optimum solution will have their cost assigned to $I_b$.
For the second trial, from the previous case argument, the expected cost share of $I_b$ is at most $\frac{3}{2} w_b$ (or at most $w_b$ if the optimum solution also queries $I_c$).
Therefore, the total expected cost share of~$I_b$ for both trials will be at most
\begin{eqnarray*}
 & & \frac{w(N_1(b))}{2 w_b} \cdot w_b + \left( 1 - \frac{w(N_1(b))}{2 w_b} \right) \left( w(N_1(b)) + \frac{3}{2} w_b \right) \\
 & = & \frac{w_b w(N_1(b))}{2 w_b} + w(N_1(b)) + \frac{3}{2} w_b - \frac{w^2(N_1(b))}{2 w_b} - \frac{3 w(N_1(b))}{4} \\
 & = & \frac{6 w_b^2 + 3 w_b w(N_1(b)) - 2 w^2(N_1(b))}{4 w_b} \cdot \frac{8}{8} \\
 & = & \frac{57 w_b^2 - 9 w_b^2 + 24 w_b w(N_1(b)) - 16 w^2(N_1(b))}{32 w_b} \leq \frac{57}{32} w_b,
\end{eqnarray*}
where the inequality uses the fact that $(3 w_b - 4 w(N_1(b)))^2 \geq 0$.
The expected cost share for $j \in N_1(b)$ such that $I_j$ is the optimum solution will be at most $w_j$.\qed
\end{enumerate}
\end{proof}

The analysis is tight if we consider a path with three intervals $I_a, I_b, I_c$ with $\delta = 0$, $w_a = 1$, $w_b = w_c = 4/3$, $v_a \in I_a \setminus I_b$, $v_b \in I_b \setminus (I_a \cup I_c)$ and $v_c \in I_b$.
Note, however, that this does not give us an improved lower bound, since for paths of length 3 we can use a similar strategy as that for 2-components, doing a randomized trial that considers querying either $I_b$, or $I_a$ and~$I_c$.

However, the analysis indicates us that there is room for improvement if we change the probabilities a little bit.
The probabilities were chosen to guarantee that the algorithm performs well for a 2-component, no matter if $I_a$ or $I_b$ is the best option.
The example in the previous paragraph shows that we cannot improve the bound when $I_b$ is the best option, but we may use a higher probability for $I_b$, thus incurring some loss in the bound when $I_a$ is the best option, but improving the bound when $I_b$ is the best option.
We claim that the ratio is minimum when $p = \min \left( 1, \displaystyle\frac{w(N_1(b))}{w_b\sqrt{3}} \right)$, in which case the query-competitive factor is improved to $1 + \frac{4}{3\sqrt{3}} \approx 1.7698$.

\begin{theorem}
Algorithm~\ref{alg:randcost} has query-competitive ratio $1 + \frac{4}{3\sqrt{3}}$ for arbitrary query costs if we replace $p = \min \left( 1, \displaystyle\frac{w(N_1(b))}{w_b\sqrt{3}} \right)$ in Line~11.
\end{theorem}

\begin{proof}
We only need to reanalyze the case when the dependency graph is a caterpillar.
Let~$P$ and $b$ be as defined in Line~10.
If $w(N_1(b)) \geq w_b\sqrt{3}$, then the algorithm deterministically queries $I_b$ first, and spends at most $1 + \frac{1}{\sqrt{3}} \approx 1.578$ times the portion of the optimum solution contained in $N_1(b) \cup \{b\}$.
If $w(N_1(b)) < w_b\sqrt{3}$, then we consider two cases.

\begin{enumerate}
 \item If the optimum solution does not query $I_b$, then the expected total cost share for intervals in $N_1(b)$ will be at most
\begin{displaymath}
 \left( 1 - \frac{w(N_1(b))}{w_b\sqrt{3}} \right) \cdot w(N_1(b)) + \frac{w(N_1(b))}{w_b\sqrt{3}} \cdot (w_b + w(N_1(b)))
 = \left( 1 + \frac{1}{\sqrt{3}} \right) \cdot w(N_1(b)).
\end{displaymath}

 \item If the optimum solution queries $I_b$, then the total expected cost share of $I_b$ for both trials will be at most
\begin{eqnarray*}
 & & \frac{w(N_1(b))}{w_b \sqrt{3}} \cdot w_b + \left( 1 - \frac{w(N_1(b))}{w_b \sqrt{3}} \right) \left( w(N_1(b)) + \left( 1 + \frac{1}{\sqrt{3}} \right) \cdot w_b \right) \\
 & = & \frac{(3 + \sqrt{3}) w_b^2 + 2 w_b w(N_1(b)) - \sqrt{3} w^2(N_1(b))}{3 w_b} \cdot \frac{\sqrt{3}}{\sqrt{3}} \\
 & = & \frac{(4 + 3\sqrt{3}) w_b^2 - w_b^2 + 2 \sqrt{3} w_b w(N_1(b)) - 3 w^2(N_1(b))}{3 \sqrt{3} w_b} \leq \left( 1 + \frac{4}{3\sqrt{3}} \right) \cdot w_b,
\end{eqnarray*}
where the inequality uses the fact that $(w_b - \sqrt{3}w(N_1(b)))^2 \geq 0$,
and the expected cost share for $j \in N_1(b)$ such that $I_j$ is the optimum solution will be at most $w_j$.
\qed
\end{enumerate}
\end{proof}

The analysis is tight if we consider a path with three intervals $I_a, I_b, I_c$ with $\delta = 0$, $w_a = 1$, $w_b = w_c = \sqrt{3}$, $v_a \in I_a \setminus I_b$, $v_b \in I_b \setminus (I_a \cup I_c)$ and $v_c \in I_b$.

Now let us prove that this is the best possible factor that can be obtained for this framework.
We are using a probability of the form $p = \min \left( 1, \displaystyle\frac{w_a}{w_b} (\alpha - 1) \right)$.
When the optimal choice is to query $I_a$, we obtain a guarantee in the form
\begin{displaymath}
\frac{\left( 1 - \displaystyle\frac{w_a}{w_b} (\alpha - 1) \right) \cdot w_a + \displaystyle\frac{w_a}{w_b} (\alpha - 1) \cdot (w_b + w_a)}{w_a} = \alpha.
\end{displaymath}
When the optimal choice is to query $I_b$, this gives us a guarantee of
\begin{displaymath}
\frac{\displaystyle\frac{w_a}{w_b} (\alpha - 1) \cdot w_b + \left( 1 - \displaystyle\frac{w_a}{w_b} (\alpha - 1) \right) ( w_a + \alpha \cdot w_b )}{w_b}.
\end{displaymath}
Let us define this latter guarantee as a function 
\begin{displaymath}
\gamma(x, y) = \frac{\displaystyle\frac{x}{y} (\alpha - 1) \cdot y + \left( 1 - \displaystyle\frac{x}{y} (\alpha - 1) \right) ( x + \alpha y )}{y}
 = \frac{x^2}{y^2} (1 - \alpha) + \frac{x}{y} \cdot \alpha(2 - \alpha) + \alpha.
\end{displaymath}
We want to find $1 < \alpha < 2$ that minimizes the maximum of $\gamma(x, y)$ for all $x,y > 0$.
Since $y > 0$, we have a critical point when
\begin{displaymath}
\left\{
 \begin{array}{l}
  \displaystyle\frac{\partial \gamma}{\partial x} = \displaystyle\frac{2x}{y^2}(1 - \alpha) + \displaystyle\frac{\alpha(2 - \alpha)}{y} = 0 \\
  \displaystyle\frac{\partial \gamma}{\partial y} = \displaystyle\frac{2x^2}{y^3}(\alpha - 1) + \displaystyle\frac{x\alpha(\alpha-1)}{y^2} = 0 \\
 \end{array}
\right. \Rightarrow 2x(1 - \alpha) + y\alpha(2 - \alpha) = 0
\Rightarrow \frac{x}{y} = \frac{\alpha(2 - \alpha)}{2(\alpha - 1)}.
\end{displaymath}
In that case, we have that
\begin{displaymath}
 \gamma(\alpha) = \left( \frac{\alpha(2 - \alpha)}{2(\alpha - 1)} \right)^2 \cdot (1-\alpha) + \left( \frac{\alpha(2 - \alpha)}{2(\alpha - 1)} \right) \cdot \alpha (2 - \alpha) + \alpha
 = \frac{\alpha^2 (2-\alpha)^2}{4(\alpha-1)}+\alpha,
\end{displaymath}
whose critical points have
\begin{displaymath}
 \frac{d\gamma}{d\alpha} = \frac{(\alpha^2 - 2\alpha + 2)(3\alpha^2 - 6 \alpha + 2)}{4(\alpha-1)^2} = 0,
\end{displaymath}
so $\alpha = 1 - \frac{1}{\sqrt{3}}$ or $\alpha = 1 + \frac{1}{\sqrt{3}}$.
Since we are looking for $1 < \alpha < 2$, we stick with the latter.

\section{Adaptive Algorithms with Queries Returning Intervals}
\label{sec:cpcp}

In this section, we investigate a variant of the problem in which a query may not necessarily return the exact value, but instead may return a more refined interval.
We do not make any assumption on how more refined the queried interval is; even though this seems too arbitrary, the adversary also has to deal with the same uncertainty, so we can indeed devise competitive algorithms.
We may even allow a query to return the same interval as before, since this is equivalent to having a very small improvement.
We must assume, however, that the relative order between two dependent intervals can be solved after a finite number of queries, which may yield an output size that is super-polynomial in the input size.
This model was proposed in \cite{gupta16queryselection}.
For this problem, we assume what they call the CP-CP model, in which the input consists of closed intervals and points, and a query can return closed intervals and points as well; in \cite{gupta16queryselection} it was shown that this assumption is no more restricted than if we also allow open intervals.

This problem has a lower bound of $2$ on the query-competitive ratio, even for randomized algorithms.
This is because in~\cite[Section~9]{megow17mst} it was proven that, if queries are allowed to return intervals, then any randomized algorithm that decides the relative order between two intervals has expected query-competitive ratio at least $2$.
We show that this lower bound applies even if query costs are uniform and the dependency graph is a proper interval graph with large components.
We present the result for deterministic algorithms, but it can be adapted for randomized algorithms using the same idea.

\begin{lemma}
\label{lemma:lowercpcp}
Any deterministic adaptive algorithm for the CP-CP model has query-competitive ratio at least~$2$, even if $\delta = 0$ and the dependency graph is a proper interval graph with large components.
\end{lemma}

\begin{proof}
Consider a path with $2n$ intervals with uniform query cost.
The first $2n-2$ intervals turn into a point after the first query, and we force them all to be in any solution by setting $v_{2i-1} \in I_{2i}$ and $v_{2i} \in I_{2i-1}$, for $i = 1, \ldots, n-1$.
For the last two intervals, we claim that we can make the optimum solution solve this pair by performing $M$ queries, while any deterministic algorithm has to perform $2M$ queries, for any integer $M > 0$, so we can make the query-competitive ratio approach $2$ if $M \gg n$.

The argument is similar to the tight example in~\cite[Section~5.1]{gupta16queryselection}.
Suppose that the algorithm has already done $2M-1$ queries.
For the first $M-1$ queries on $I_{2n-1}$, we obtain the same interval; the same applies to $I_{2n}$.
Then we have two cases, depending on whether the algorithm makes $M$ queries on $I_{2n-1}$ or on $I_{2n}$.
If the algorithm makes at least $M$ queries on $I_{2n-1}$, then we return the same interval for all subsequent queries on $I_{2n-1}$, and we return $v_{2n} \in I_{2n} \setminus I_{2n-1}$ in the $M$-th query on~$I_{2n}$.
Thus the algorithm has to make $2M$ queries, and the optimum solution can simply query $M$ times $I_{2n}$.
The argument is symmetric if the algorithm makes at least $M$ queries on $I_{2n}$.
\qed
\end{proof}

Now we give a 2-query-competitive deterministic algorithm for this version of the problem.
We also cover the case in which query costs change over time, i.e., we assume that querying interval $I_i$ for the $t$-th time costs $w_i(t) \in \reais$, for $t = 1, 2, \ldots$.
The algorithm is a simple modification of the local ratio 2-approximation algorithm for the vertex cover problem~\cite{baryehuda81vc}.
It is also a generalization of both algorithms in Section~\ref{sec:adaptive}.
We begin by querying intervals whose current query cost is zero, since this does not affect the solution cost.
If there is some dependency between vertices $I_i$ and $I_j$, then we subtract from their current query cost the minimum of them; this will force one of them to be queried.
We query intervals that are singleton witness sets according to Fact~\ref{fact:offnecessary}, and proceed until all dependencies are resolved.
A pseudocode is presented in Algorithm~\ref{alg:cpcp}.

\begin{algorithm}[!ht]
\SetAlgoNoEnd
\KwIn{$(I_1, \ldots, I_n, w)$}
\For{$i \recebe 1$ \KwTo $n$}{
 $t_i \recebe 1$;
}
\While{there is some dependency}{
 \eIf{some $I_i$ has $w_i(t_i) = 0$}{
  query $I_i$\;
  $t_i \recebe t_i +1$\;
 }{
  \KwLet $I_i$ and $I_j$ be two dependent intervals\;
  $W \recebe \min \{w_j(t_j), w_i(t_i)\}$\;
  $w_i(t_i) \recebe w_i(t_i) - W$\;
  $w_j(t_j) \recebe w_j(t_j) - W$\;
 }
 \While{there are $i,j$ with $I_i \supset [\ell_j - \delta, r_j + \delta]$}{
  query $I_i$\;
  $t_i \recebe t_i +1$\;
 }
}
\vspace{0.2cm}
\caption{\label{alg:cpcp} Adaptive algorithm for queries returning intervals.}
\end{algorithm}

\begin{theorem}
Algorithm~\ref{alg:cpcp} is $2$-query-competitive for the sorting problem with uncertainty in the CP-CP model, even if query costs change over time.
\end{theorem}

\begin{proof}
The proof is by induction and relies on the Local Ratio Theorem (Theorem~\ref{teo:localratio}).
When we query an interval with zero query cost in Line 5, it does not affect the solution cost.
Intervals that are queried in Line~13 must be in any solution, due to Fact~\ref{fact:offnecessary}.
If we run Lines 8--11 for intervals $I_i$ and $I_j$ for given $t_i$ and~$t_j$, we apply the Local Ratio Theorem with $w^{(1)}_i(t_i) = w^{(1)}_j(t_j) = W$, $w^{(1)}_k(t') = 0$ for $(k, t') \notin \{ (i, t_i), (j, t_j) \}$, and $w^{(2)} = w - w^{(1)}$.
By induction hypothesis, the returned solution is $2$-query-competitive on $w^{(2)}$.
The pair~$ij$ is not resolved before we make $t_i$ queries in $I_i$ and $t_j$ queries in $I_j$, and due to Lemma~\ref{lemma:decideind} we must query at least one of them.
In the worst case the algorithm will query both of them, so the returned solution is also $2$-query-competitive on~$w^{(1)}$ and, by the Local Ratio Theorem, it is $2$-query-competitive on $w$.
\qed
\end{proof}

This algorithm also works for the following generalization of the vertex cover problem: suppose we have an arbitrary graph, and we want to resolve all the edges.
Querying a vertex may resolve an incident edge, but for some edges it may be necessary to query both endpoints.
We do not know this information, so we have an online problem.
The algorithm works even if vertex weights are not uniform, and if a vertex may be required to be queried multiple times before the edge is resolved.

It is interesting to note that we obtain a factor of $2$, despite the fact that there are instances for which the $2$-approximation guarantee for the vertex cover problem is tight~\cite{baryehuda04localratio}, and that in the proof of Theorem~\ref{teo:detarbitrary} we argue that we pay the cost of a minimum vertex cover plus the cost of the remaining dependent intervals.
This indicates that in some cases the minimum vertex cover is a loose lower bound to the optimum solution.

\section{Advice Complexity for Adaptive Algorithms}
\label{sec:advice}

In this section we investigate the advice complexity of solving the adaptive version of the problem.
Recall that the advice complexity is the number of bits of advice from an oracle that are sufficient and necessary for an online algorithm to solve the problem exactly.
We assume arbitrary query costs, and for consistency that the oracle answers questions regarding a fixed optimum solution for the given instance.

First, we deal with the case when $\delta = 0$.
Let $n$ be the number of given intervals.
We claim that $\lfloor n/2 \rfloor$ bits of advice are sufficient to solve the problem exactly, and that there are instances for which $\lfloor n/2 \rfloor$ bits are necessary.

\begin{lemma}
The advice complexity of the adaptive sorting problem with uncertainty is at least $\lfloor n/2 \rfloor$, where~$n$ is the number of intervals, even if $\delta = 0$.
\end{lemma}

\begin{proof}
Assume $n$ even and consider $n/2$ independent copies of the bad instance of Lemma~\ref{lemma:loweradap}.
At least 1 bit of advice is necessary to decide the relative order between each pair.
\qed
\end{proof}

For an adaptive algorithm with a matching upper bound, we note that, if $\delta = 0$, then any triangle $ijk$ contains a vertex $j$ such that $I_j \subseteq I_i \cup I_k$ (just take~$i$ with minimum $\ell_i$ and $k$ with maximum~$r_k$).
Thus, we can ask the oracle whether the optimum solution queries $I_j$; if not, then we must query all neighbors of~$j$; otherwise, we query $I_j$, and since $I_j \subseteq I_i \cup I_k$, we will know at least one of $I_i$ and~$I_k$ that also must be queried.
If the dependency graph contains no triangles, then it is a forest, because any cycle in a chordal graph must contain a triangle.
Therefore, we can pick a leaf $i$ and ask the oracle whether the optimum solution queries its neighbor $j$; if not, then we query all neighbors of $j$; otherwise, we query~$I_j$ and we will know whether $I_i$ must be queried or not.
Since we decide at least two intervals with one bit of advice, then $\lfloor n/2 \rfloor$ bits are sufficient.
We present a pseudocode in Algorithm~\ref{alg:advice2}.

\begin{algorithm}[!ht]
\SetAlgoNoEnd
\KwIn{$(I_1, \ldots, I_n)$}
$\Vcal \recebe \emptyset$\;
\While{there is some dependency}{
 \uIf{there is a triangle $K$}{
  \KwLet $i \in K$ with minimum $\ell_i$, $k \in K$ with maximum $r_k$, and $j \in K \setminus \{i, k\}$\;
 }
 \lElse{
  \KwLet $i$ be a leaf, and $j$ be the neighbor of $i$
 }
 ask the oracle whether the optimum solution queries $j$\;
 \lIf{yes}{query $I_j$, add $v_j$ to $\Vcal$}
 \ElseForEach{neighbor $z$ of $j$}{
   query $I_z$, add $v_z$ to $\Vcal$\;
 }
 \While{there is $I_i \supset [v_j - \delta, v_j + \delta]$ for some $v_j \in \Vcal$}{
  query $I_i$, add $v_i$ to $\Vcal$\;
 }
}
\vspace{0.2cm}
\caption{\label{alg:advice2} An adaptive algorithm that finds an optimum solution with $\lfloor n/2 \rfloor$ bits of advice when $\delta = 0$.}
\end{algorithm}

\begin{theorem}
The advice complexity of the adaptive sorting problem with uncertainty is $\lfloor n/2 \rfloor$ when $\delta = 0$, where $n$ is the number of intervals.
\end{theorem}

Now we consider the case when $\delta > 0$.
Here, we can improve the lower bound to $\lceil n/3 \cdot \lg 3 \rceil$ and still have an algorithm with matching upper bound.
Both are based on the fact that to encode $k$ distinct values amortized $\lg k$ bits are sufficient and necessary~\cite{shannon48com}.

\begin{lemma}
The advice complexity of the adaptive sorting problem with uncertainty is at least $\lceil n/3 \cdot \lg 3 \rceil$, where~$n$ is the number of intervals.
\end{lemma}

\begin{proof}
Assume $n$ multiple of $3$ and consider $n/3$ independent triangles; it suffices to bound the number of bits of advice necessary to solve each triangle.
Suppose by contradiction that there is an algorithm that solves any triangle with one bit of advice, and consider the following instances $\Ical_1, \Ical_2, \Ical_3$.
In each $\Ical_i$, the $k$-th triangle has intervals $I_1, I_2, I_3$ such that $\ell_1 < \ell_2 < \ell_3 - \delta$, $r_1 + \delta < r_2 < r_3$, $\ell_2 \leq \ell_1 + \delta$, $r_2 \geq r_3 - \delta$ and $r_1 - \ell_3 > 2\delta$.
We have $v_j = r_j$ for all $j$ in $\Ical_1$; in $\Ical_2$, $v_1 = \ell_1$, $v_2 \in (\ell_3 + \delta, r_1 - \delta)$, $v_3 = r_3$; and $v_j = \ell_j$ for all $j$ in $\Ical_3$.
The only optimum solution for $\Ical_1, \Ical_2, \Ical_3$ is not to query $I_1, I_2, I_3$, respectively.
(See Figure~\ref{fig:loweradv}.)

\begin{figure}[!ht]
  \centering
  \subfigure[]{
   \tikzstyle{every node}=[circle, draw, fill=black, inner sep=0pt, minimum width=4pt]
   \begin{tikzpicture}[thick, scale=0.45]
    \draw (0, 0) -- (7, 0) node{};
    \draw (1, 0)[ultra thick] -- (6, 0);
    \draw (0.5, -1) -- (9, -1) node{};
    \draw (1.5, -1)[ultra thick] -- (8, -1);
    \draw (2.5, -2) -- (9.5, -2) node{};
    \draw (3.5, -2)[ultra thick] -- (8.5, -2);
   \end{tikzpicture}
  }\qquad
  \subfigure[]{
   \tikzstyle{every node}=[circle, draw, fill=black, inner sep=0pt, minimum width=4pt]
   \begin{tikzpicture}[thick, scale=0.45]
    \draw (0, 0) node{} -- (7, 0);
    \draw (1, 0)[ultra thick] -- (6, 0);
    \draw (0.5, -1) -- (9, -1) node[midway]{};
    \draw (1.5, -1)[ultra thick] -- (8, -1);
    \draw (2.5, -2) -- (9.5, -2) node{};
    \draw (3.5, -2)[ultra thick] -- (8.5, -2);
   \end{tikzpicture}
  }\qquad
  \subfigure[]{
   \tikzstyle{every node}=[circle, draw, fill=black, inner sep=0pt, minimum width=4pt]
   \begin{tikzpicture}[thick, scale=0.45]
    \draw (0, 0) node{} -- (7, 0);
    \draw (1, 0)[ultra thick] -- (6, 0);
    \draw (0.5, -1) node{} -- (9, -1);
    \draw (1.5, -1)[ultra thick] -- (8, -1);
    \draw (2.5, -2) node{} -- (9.5, -2);
    \draw (3.5, -2)[ultra thick] -- (8.5, -2);
   \end{tikzpicture}
  }
  \caption{Instances for the lower bound on advice complexity when $\delta > 0$.}
  \label{fig:loweradv}
\end{figure}
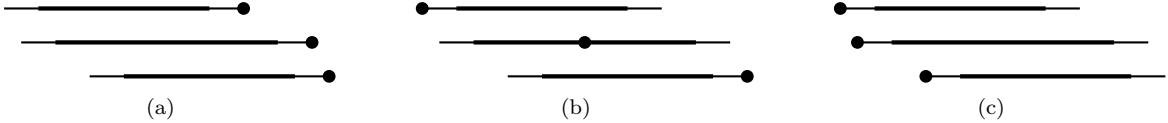

By the pigeonhole principle, the algorithm must have the same advice for at least two of those inputs.
So, it suffices to prove that any deterministic algorithm fails in one instance of any subset with at least two of those instances.
Since the intervals are structurally identical, any algorithm for a triangle performs no better than an algorithm in the following form, for fixed $x, y \in \{1, 2, 3\}$, $x \neq y$: query $I_x$, and if no helpful information is given, query~$I_y$.
The instances are constructed in such a way that, for instance $\Ical_i$, the algorithm does not get any helpful information by querying~$I_x$ with $i \neq x$, so it fails on instances $\Ical_x$ and $\Ical_y$.
Since one bit is not sufficient, at least three different values must be encoded in the advice for each triangle, so $\lceil n/3 \cdot \lg 3 \rceil$ bits are necessary for the whole instance.
\qed
\end{proof}

The algorithm that attains the upper bound relies on Lemma~\ref{lemma:simplicial}.
It considers the clique~$K$ consisting of a non-isolated vertex $x$ with minimum $r_x$ and its neighborhood.
Then it asks the oracle for the index of a vertex $y$ in $K$ that is not queried in the optimum solution or, if there is no such vertex in $K$, then the oracle must return $y = x$.
Either way, the algorithm queries all intervals in $K \setminus \{y\}$, and if $y = x$ then the algorithm will know if $y$ must also be queried after querying everyone else.
So it uses $\lg |K|$ bits of advice to decide at least~$|K|$ intervals, and the bound follows since $\lg k / k$ has its maximum at $k = 3$ when $k$ is integer.
A pseudocode is presented in Algorithm~\ref{alg:advicelg3}.

\begin{algorithm}[!ht]
\SetAlgoNoEnd
\KwIn{$(I_1, \ldots, I_n)$}
$\Vcal \recebe \emptyset$\;
\While{there is some dependency}{
 \KwLet $x$ non-isolated with minimum $r_x$, and $K$ be the clique consisting of $x$ and its neighborhood\;
 ask the oracle for a vertex $y \in K$ not queried in the optimum solution, or $y = x$ if there is no such vertex\;
 \ForEach{$z \in K \setminus \{y\}$}{
  query $I_z$, add $v_z$ to $\Vcal$\;
 }
 \While{there is $I_i \supset [v_j - \delta, v_j + \delta]$ for some $v_j \in \Vcal$}{
  query $I_i$, add $v_i$ to $\Vcal$\;
 }
}
\vspace{0.2cm}
\caption{\label{alg:advicelg3} An adaptive algorithm that finds an optimum solution with $\lceil n/3 \cdot \lg 3 \rceil$ bits of advice.}
\end{algorithm}

\begin{theorem}
The advice complexity of the adaptive sorting problem with uncertainty is $\lceil n/3 \cdot \lg 3 \rceil$, where~$n$ is the number of intervals.
\end{theorem}

\section{Towards Characterizing co-TT Graphs by Forbidden Induced Subgraphs}
\label{sec:charac}

In this section we discuss the importance of our sorting problem with uncertainty for the understanding of the class of co-TT graphs itself.
A first point is that we are not aware of other applications of this graph class in the literature.
Second, there are various characterizations of co-TT graphs in the literature~\cite{monma88ttolerance,golumbic14cott}, and they can be recognized in $\Oh(n^2)$ time~\cite{golovach17tt}, but a characterization in terms of forbidden induced subgraphs is an open question.

Before we figured out that Definition~\ref{def:dep} is equivalent to that of co-TT graphs, we spent some time trying to understand the graph class we were dealing with.
Since interval graphs are an obvious subclass, and a nice characterization by forbidden induced subgraphs is known for interval graphs~\cite{lekkeikerker62interval}, our first direction was to test which of these graphs are also forbidden induced subgraphs of co-TT graphs.
It turns out we got to extend the partial list of forbidden induced subgraphs to the one presented in Figure~\ref{fig:coTT}.

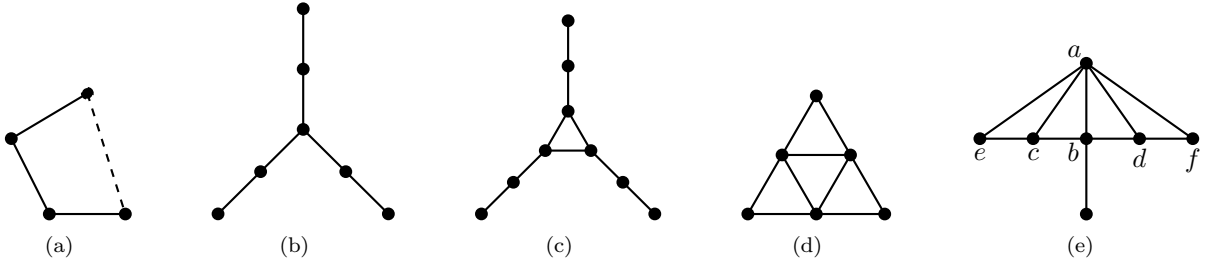
\begin{figure}[!ht]
  \centering
  \subfigure[]{\label{fig:chordal}
   \tikzstyle{every node}=[circle, draw, fill=black, inner sep=0pt, minimum width=4pt]
   \begin{tikzpicture}[thick, scale=1]
    \draw (2, 0) node{} -- (1, 0);
    \draw (1, 0) node{} -- (0.5, 1);
    \draw (0.5, 1) node{} -- (1.5, 1.6);
    \draw[dashed] (1.5, 1.6) node{} -- (2, 0);
   \end{tikzpicture}
  }\qquad
  \subfigure[]{\label{fig:asteroid1}
   \tikzstyle{every node}=[circle, draw, fill=black, inner sep=0pt, minimum width=4pt]
   \begin{tikzpicture}[thick, scale=0.8]
    \draw (0, 0) node{} -- (0, 1) node{};
    \draw (0, 1) -- (0, 2) node{};
    \draw (0, 0) -- (-0.7, -0.7) node{};
    \draw (-0.7, -0.7) -- (-1.4, -1.4) node{};
    \draw (0, 0) -- (0.7, -0.7) node{};
    \draw (0.7, -0.7) -- (1.4, -1.4) node{};
   \end{tikzpicture}
  }\qquad
  \subfigure[]{\label{fig:coTT2}
   \tikzstyle{every node}=[circle, draw, fill=black, inner sep=0pt, minimum width=4pt]
   \begin{tikzpicture}[thick, scale=0.6]
    \draw (0, 0) node{} -- (0, 1) node{};
    \draw (0, 1) -- (0, 2) node{};
    \draw (0, 0) -- (-0.5, -0.87) node{};
    \draw (-0.5, -0.87) -- (0.5, -0.87) node{};
    \draw (0, 0) -- (0.5, -0.87);
    \draw (-0.5, -0.87) -- (-1.2, -1.57) node{};
    \draw (-1.2, -1.57) -- (-1.9, -2.27) node{};
    \draw (0.5, -0.87) -- (1.2, -1.57) node{};
    \draw (1.2, -1.57) -- (1.9, -2.27) node{};
   \end{tikzpicture}
  }\qquad
  \subfigure[]{\label{fig:sun3}
   \tikzstyle{every node}=[circle, draw, fill=black, inner sep=0pt, minimum width=4pt]
   \begin{tikzpicture}[thick, scale=0.9]
    \draw (0, 0) node{} -- (-1, 0) node{};
    \draw (0, 0) -- (1, 0) node{};
    \draw (0, 0) -- (-0.5, 0.87) node{};
    \draw (0, 0) -- (0.5, 0.87) node{};
    \draw (-1, 0) -- (-0.5, 0.87);
    \draw (1, 0) -- (0.5, 0.87);
    \draw (-0.5, 0.87) -- (0.5, 0.87);
    \draw (-0.5, 0.87) -- (0, 1.74) node{};
    \draw (0.5, 0.87) -- (0, 1.74);
   \end{tikzpicture}
  }\qquad
  \subfigure[]{\label{fig:asteroid2}
   \tikzstyle{every node}=[circle, draw, fill=black, inner sep=0pt, minimum width=4pt]
   \begin{tikzpicture}[thick, scale=1]
    \draw (0, 0) node[label=south west:$b$]{} -- (-0.7, 0) node[label=below:$c$]{};
    \draw (-0.7, 0) -- (-1.4, 0) node[label=below:$e$]{};
    \draw (0, 0) -- (0.7, 0) node[label=below:$d$]{};
    \draw (0.7, 0) -- (1.4, 0) node[label=below:$f$]{};
    \draw (0, 0) -- (0, -1) node{};
    \draw (0, 0) -- (0, 1) node[label=north west:$a$]{};
    \draw (-0.7, 0) -- (0, 1);
    \draw (-1.4, 0) -- (0, 1);
    \draw (0.7, 0) -- (0, 1);
    \draw (1.4, 0) -- (0, 1);
   \end{tikzpicture}
  }
  \caption{A partial list of forbidden induced subgraphs for co-TT graphs.
  \subref{fig:chordal} is the $k$-cycle, for $k \geq 4$.
  \subref{fig:sun3} is the $3$-sun.}
  \label{fig:coTT}
\end{figure}

Figure~\ref{fig:chordal} ($k$-cycle, $k \geq 4$) is inherited from chordal graphs.
Figures~\ref{fig:asteroid1} and~\ref{fig:coTT2} have long been known not to be co-TT~\cite{golumbic84tolerance}, and Figure~\ref{fig:sun3} (the $3$-sun) is proved not to be co-TT in~\cite{calamoneri14tt}.
We prove that the graph in Figure~\ref{fig:asteroid2}, which cannot be an interval graph, cannot be a co-TT graph either.
The following facts will be useful; similar or equivalent facts have been known for the red/blue characterization of co-TT graphs~\cite{golumbic14cott}.

\begin{fact}
\label{lemma:trivialind}
Two trivial intervals cannot be dependent.
\end{fact}

\begin{fact}
\label{lemma:path2trivial}
A trivial interval is always simplicial, i.e., cannot be dependent of two independent intervals.
\end{fact}

\begin{proof}
By Fact~\ref{lemma:trivialind}, the neighbors of a trivial interval are non-trivial.
Let $I_i$ and $I_j$ be two non-trivial intervals, which are independent, and let~$I_k$ be a trivial interval.
Assume without loss of generality that $r_i - \ell_j \leq \delta$, and suppose by contradiction that $I_k$ is dependent of both~$I_i$ and $I_j$.
We have that $r_k - \ell_j > \delta$ and, since $I_k$ is trivial, $r_k - \ell_k \leq \delta$.
Thus,
$r_i \leq \ell_j + \delta < r_k \leq \ell_k + \delta$,
which contradicts the fact that $I_i$ and~$I_k$ are dependent.
\qed
\end{proof}

\begin{fact}
\label{fact:containdep}
Let $I_i$ and $I_j$ be two intervals such that $I_i \supseteq I_j$.
If $I_j$ is dependent to some interval $I_k$, then $I_i$ is also dependent to $I_k$.
\end{fact}

\begin{proof}
We have that $r_i - \ell_k \geq r_j - \ell_k > \delta$ and $r_k - \ell_i \geq r_k - \ell_j > \delta$.
\qed
\end{proof}

\begin{lemma}
\label{lemma:umbrella}
The graph of Figure~\ref{fig:asteroid2} cannot be a co-TT graph.
\end{lemma}

\begin{proof}
Since $c$ and $d$ are independent, we may assume that $r_c - \ell_d \leq \delta$.
Since $b$ and $d$ are dependent, we have that $r_b > \ell_d + \delta \geq r_c$.
The dependency between $c$ and $e$ implies that $r_b - \ell_e > r_c - \ell_e > \delta$.
Thus, we have that $r_e - \ell_b \leq \delta$, because $b$ and $e$ are independent.
Then, since $a$ and $e$ are dependent, $\ell_a < r_e - \delta \leq \ell_b$.

By symmetry, we can prove that $\ell_b < \ell_d$, thus $r_b - \ell_f \leq \delta$ and $r_a > r_b$.
Thus $I_a$ contains~$I_b$ and $b$ cannot have a neighbor that is not adjacent to $a$.
\qed
\end{proof}

In the opposite direction, we prove that the graphs in Figure~\ref{fig:asteroid}, which are forbidden for interval graphs~\cite{lekkeikerker62interval}, can occur as co-TT graphs when $k \geq 2$.
(The graph of Figure~\ref{fig:asteroid3} with $k = 2$ has long been known to be co-TT~\cite{calamoneri14tt,golumbic84tolerance}.)
In Figure~\ref{fig:int_asteroid}, we show how to realize those graphs as instances of our sorting problem with uncertainty.
In both cases, if $\ell_{b''} = x$, then we take, for some $0 < \epsilon < \delta$, $r_{b'} = x + \delta + \epsilon$ and $I_e = [x + \epsilon, x + \delta]$.
Then~$b'$ and~$b''$ are dependent because $r_{b'} - \ell_{b''} = \delta + \epsilon > \delta$ (and clearly $r_{b''} - \ell_{b'} > \delta$).
But  $r_{b'} - \ell_e = \delta$ and $r_e - \ell_{b''} = \delta$, so $e$ is dependent to neither $b'$ nor $b''$.

\begin{figure}[!ht]
 \centering
  \subfigure[]{\label{fig:asteroid3}
   \tikzstyle{every node}=[circle, draw, fill=black, inner sep=0pt, minimum width=4pt]
   \begin{tikzpicture}[thick, scale=1]
    \draw (0, 0) node[label=below:$b_3$]{} -- (-0.7, 0) node[label=below:$b_2$]{};
    \draw (-0.7, 0) -- (-1.4, 0) node[label=below:$b_1$]{};
    \draw (-1.4, 0) -- (-2.1, 0) node[label=below:$c$,fill=white]{};
    \draw[dashed] (0, 0) -- (1.4, 0) node[label=below:$b_k$]{};
    \draw (1.4, 0) -- (2.1, 0) node[label=below:$d$,fill=white]{};
    \draw (0, 0) -- (0, 1) node[label=north west:$a$]{};
    \draw (0, 1) -- (0, 2) node[label=west:$e$,fill=white]{};
    \draw (-0.7, 0) -- (0, 1);
    \draw (-1.4, 0) -- (0, 1);
    \draw (1.4, 0) -- (0, 1);
   \end{tikzpicture}
  }\qquad
  \subfigure[]{\label{fig:asteroid4}
   \tikzstyle{every node}=[circle, draw, fill=black, inner sep=0pt, minimum width=4pt]
   \begin{tikzpicture}[thick, scale=1]
    \draw (0, 0) node[label=below:$b_3$]{} -- (-0.7, 0) node[label=below:$b_2$]{};
    \draw (-0.7, 0) -- (-1.4, 0) node[label=below:$b_1$]{};
    \draw (-1.4, 0) -- (-2.1, 0.85);
    \draw[dashed] (0, 0) -- (1.4, 0) node[label=below:$b_k$]{};
    \draw (1.4, 0) -- (2.1, 0.85);
    \draw (0, 0) -- (-0.8, 1) node[label=north west:$a$]{};
    \draw (0, 0) -- (0.8, 1) node[label=north east:$a'$]{};
    \draw (-0.8, 1) -- (0, 2);
    \draw (0.8, 1) -- (0, 2) node[label=west:$e$,fill=white]{};
    \draw (-0.7, 0) -- (-0.8, 1);
    \draw (-1.4, 0) -- (-0.8, 1);
    \draw (-0.7, 0) -- (0.8, 1);
    \draw (-1.4, 0) -- (0.8, 1);
    \draw (1.4, 0) -- (-0.8, 1);
    \draw (1.4, 0) -- (0.8, 1);
    \draw (-0.8, 1) -- (0.8, 1);
    \draw (-0.8, 1) -- (-2.1, 0.85) node[label=north west:$c$,fill=white]{};
    \draw (0.8, 1) -- (2.1, 0.85) node[label=north east:$d$,fill=white]{};
   \end{tikzpicture}
  }
  \caption{Two families of graphs that cannot occur as interval graphs~\cite{lekkeikerker62interval}.
  In~\subref{fig:asteroid3} we have $k \geq 2$.
  In~\subref{fig:asteroid4} we have $k \geq 1$.
  Both families of graphs are co-TT graphs when $k \geq 2$.
  The white vertices are the only ones that can be trivial, if we take into account Facts~\ref{lemma:trivialind} and \ref{lemma:path2trivial}.}
  \label{fig:asteroid}
\end{figure}
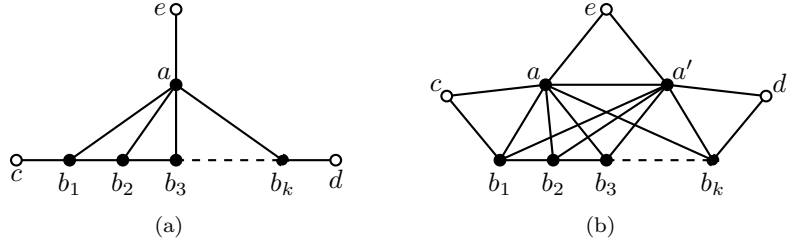

\begin{figure}[!ht]
  \centering
  \subfigure[]{\label{fig:int_asteroid3}
   \begin{tikzpicture}[thick, scale=0.6]
    \draw (0, 0) -- (3, 0) node[midway,above]{$c$};
    \draw (4, 0) -- (14, 0) node[midway,above]{$a$};
    \draw (15, 0) -- (18, 0) node[midway,above]{$d$};
    \draw (1, -1) -- (6.2, -1) node[midway,above]{$b_1$};
    \draw (7, -1) -- (11, -1) node[midway,above]{$b_3$};
    \draw (12.3, -1) node[anchor=east]{$\cdots$} -- (17, -1) node[midway,above]{$b_k$};
    \draw (4.8, -2) -- (9, -2) node[midway,above]{$b_2$};
    \draw (5.1, -3) -- (5.9, -3) node[midway,above]{$e$};
    \draw[dotted] (4.8, -0.8) -- (4.8, -3.5) node[anchor=east]{$x$};
    \draw[dotted] (5.1, -0.8) -- (5.1, -4.2) node[anchor=east]{$x + \epsilon$};
    \draw[dotted] (5.9, -0.8) -- (5.9, -4.2) node[anchor=west]{$x + \delta$};
    \draw[dotted] (6.2, -0.8) -- (6.2, -3.5) node[anchor=west]{$x + \delta + \epsilon$};
   \end{tikzpicture}
  }\qquad
  \subfigure[]{\label{fig:int_asteroid4}
   \begin{tikzpicture}[thick, scale=0.6]
    \draw (0, 1) -- (3, 1) node[midway,above]{$c$};
    \draw (4, 1) -- (17, 1) node[midway,above]{$a'$};
    \draw (1, 0) -- (14, 0) node[midway,above]{$a$};
    \draw (15, 0) -- (18, 0) node[midway,above]{$d$};
    \draw (1, -1) -- (6.2, -1) node[midway,above]{$b_1$};
    \draw (7, -1) -- (11, -1) node[midway,above]{$b_3$};
    \draw (12.3, -1) node[anchor=east]{$\cdots$} -- (17, -1) node[midway,above]{$b_k$};
    \draw (4.8, -2) -- (9, -2) node[midway,above]{$b_2$};
    \draw (5.1, -3) -- (5.9, -3) node[midway,above]{$e$};
    \draw[dotted] (4.8, -0.8) -- (4.8, -3.5) node[anchor=east]{$x$};
    \draw[dotted] (5.1, -0.8) -- (5.1, -4.2) node[anchor=east]{$x + \epsilon$};
    \draw[dotted] (5.9, -0.8) -- (5.9, -4.2) node[anchor=west]{$x + \delta$};
    \draw[dotted] (6.2, -0.8) -- (6.2, -3.5) node[anchor=west]{$x + \delta + \epsilon$};
   \end{tikzpicture}
  }
  \caption{\subref{fig:int_asteroid3} A realization of the graph of Figure~\ref{fig:asteroid3}.
  \subref{fig:int_asteroid4} A realization of the graph of Figure~\ref{fig:asteroid4} when $k \geq 2$.}
  \label{fig:int_asteroid}
\end{figure}
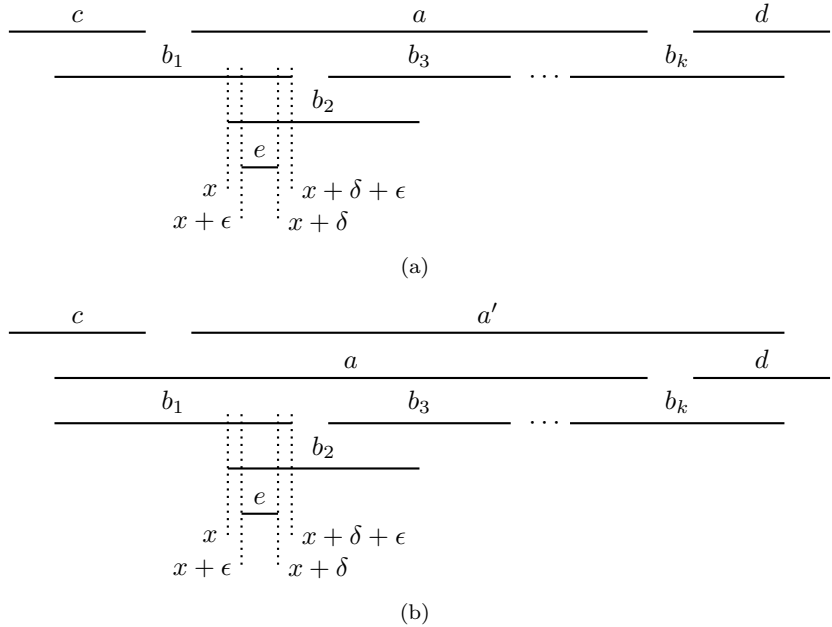

\section{Interval Problems with Uncertainty}
\label{sec:interval}

In this section, we discuss uncertainty variants of some classical problems on intervals.
In those variants, the boundary of each interval $I_i$ is given by uncertainty intervals $L_i = [\ell_{L_i}, r_{L_i}]$ and $R_i = [\ell_{R_i}, r_{R_i}]$.
We denote the precise lower and upper bounds of $I_i$ by $\ell_i$ and $r_i$, respectively, which are initially unknown and can be learned by querying $L_i$ and $R_i$, respectively.
Those were the problems we started to investigate in the model of uncertainty optimization with queries, and from them we got inspiration to work on the sorting problem of the previous sections.

We begin with the problem of finding a maximum independent set of intervals.
The problem has query-competitive ratio at least $n-1$, even if query costs are uniform and the lower bound $L_i$ is trivial for every interval $I_i$.
(The same applies if the upper bounds are trivial instead of the lower bounds.)
To prove this, consider an interval $I_n = [\ell_n, r_n]$ with trivial lower and upper bounds, and $n-1$ identical intervals, with $\ell_{R_i} < \ell_n < r_{R_i}$ and $r_n > r_{R_i}$, for $i = 1, \ldots, n-1$.
Clearly we can have at most 2 independent intervals.
For the first $n-2$ queries, if the algorithm queries interval $I_i$, then the adversary chooses $r_i = r_{R_i}$, so $I_i$ and $I_n$ are dependent.
Then, the adversary chooses $r_j = \ell_{R_j}$ for the $(n-1)$-th queried interval $I_j$.
Clearly it would suffice to query~$I_j$ to know that there exists an independent set of size 2.

We also consider the problem of find the stabbing number of a set of intervals.
The stabbing number is the size of a minimum set of points $P$, such that each interval contains some point in~$P$.
(I.e., we wish to find a minimum transversal for a set of intervals.)
This is equivalent to finding a minimum covering by cliques, which can be solved in polynomial time~\cite{gavril72chordal}.
In the uncertainty version, this problem also has competitive ratio at least $n-1$, even if query costs are uniform and lower bounds (upper bounds) are trivial.
We use the same set of intervals as in the bad instance for the maximum independent set problem.
Clearly the stabbing number is either 1 or 2.
For the first $n-2$ queries, if the algorithm queries interval~$I_i$, then the adversary chooses $r_i = r_{R_i}$, so 1 stabbing point is sufficient for now.
Then, the adversary chooses $r_j = \ell_{R_j}$ for the $(n-1)$-th queried interval $I_j$.
Clearly it would suffice to query~$I_j$ to know that 2 stabbing points are necessary.

\bibliographystyle{plainurl}
\bibliography{../doutorado}

\end{document}